\documentclass[12pt]{article}
\usepackage{amsmath,amsfonts, amssymb,thm-restate}
\usepackage{natbib}
\usepackage{epsfig}
\usepackage[singlespacing]{setspace}
\usepackage{graphicx}
\usepackage{color}
\usepackage{float}
\usepackage{cases}
\usepackage{algorithm}
\usepackage[noend]{algpseudocode}

\usepackage{tikz}
\tikzstyle{vertex}=[circle, draw, inner sep=1pt, minimum size=15pt]
\usetikzlibrary{shapes}
\usetikzlibrary{arrows.meta}
\tikzset{>={Latex[width=2mm,length=2mm]}}

\usepackage{caption}
\usepackage{subcaption}

\PassOptionsToPackage{hyphens}{url}\usepackage{hyperref}

\newtheorem{theorem}{\sc {Theorem}}[section]
\newtheorem{lemma}{\sc {Lemma}}[section]

\newtheorem{example}{\sc {Example}}

\newtheorem{definition}{\sc{Definition}}[section]
\newtheorem{proposition}{\sc {Proposition}}[section]
\newtheorem{observation}{\sc{Observation}}[section]
\DeclareMathOperator*{\maximize}{max}
\DeclareMathOperator*{\minimize}{min}

\newenvironment{proof}[1][\bf{Proof.}]{\begin{trivlist}
\item[\hskip \labelsep {\bfseries #1}]}{\end{trivlist}}

\newenvironment{remark}[1][\bf{Remark.}]{\begin{trivlist}
\item[\hskip \labelsep {\bfseries #1}]}{\end{trivlist}}

\newcommand{\QED}{\hfill \mbox{\rule{2mm}{2mm}} \vspace{0.5cm}}

\newcommand{\Xomit}[1]{}

\usepackage[margin=1in]{geometry}

\begin{document}
\title{Market Equilibrium in Multi-tier Supply Chain Networks}
 \author{Tao Jiang\footnote{Krannert School of Management, Purdue University, taujiang300@gmail.com} \hspace{5mm}    Young-San Lin \footnote{Department of Computer Science, Purdue University, lin532@purdue.edu}\hspace{5mm} Th\`anh Nguyen\footnote{Krannert School of Management, Purdue University, nguye161@purdue.edu}   }
\maketitle
\onehalfspacing
\begin{abstract}
We consider a sequential decision model over multi-tier supply chain networks and show that in particular, for series parallel networks, there is a unique equilibrium. We provide a linear time algorithm to compute the equilibrium and study the impact and invariant of the network structure to the total trade flow and social welfare. 
\end{abstract}


\section{Introduction} \label{sec: intro}

Supply chain networks in practice are multi-tier and heterogeneous. A firm's  decision influences not only other firms within the same tier but also across. The literature on game theoretical models of supply chain networks, however,  has largely focused on two extreme cases:  heterogeneous 2-tier networks (bipartite graph) \cite{kranton2001theory} and \cite{bimpikis2014cournot} and a linear chain of $n$-tier firms \cite{wright2014buyers} and \cite{nguyen2017local}.    One main reason for this is that most models of sequential decision making in multi-tier supply chain networks are intractable. Sequential decision making is a well-observed phenomenon in supply chains because firms  at the top tier typically need to make decisions on the quantity and the price to sell to firms in the next tier and the buying firms then decide how much to buy from which suppliers, and continue to pass on the goods by determining the quantity and price for firms at the next level.  
 
To study such models, one needs to analyze subgame perfect equilibria, where each firm internalizes the decisions of all the firms downstream and compete with all the firms of the same tier. Another factor that further complicates models of general supply chain networks is that even the basic concept of tiers is ambiguous because there are often multiple routes of different length that goods are traded from the original producers to the consumers. Our paper studies a model of sequential network game motivated by supply chain network applications. Our main goal is to understand the existence of an equilibrium and the effect of network structure on the efficiency of the system. 

The length and the number of trading routes are the two main factors that impact the efficiency of a supply chain network. On one hand, a large variety of  options to trade indicates a high degree of competition. On the other hand, a long trading path causes double, triple and higher degree marginalization problems. To capture these ideas, we consider a sequential game theoretical model for a special class of networks, the {\em series parallel graphs}. We focus our analysis on  these networks because they are rich enough for studying the trade-off described above and simple enough for  characterizing the equilibrium outcomes. In particular, series parallel networks have two natural compositions. A parallel composition, which merges two different sub-networks at the source and the sink, can capture the increase in competition. A series composition, which attaches two sub-networks sequentially, corresponds to the increase in the length of trading.

\subsection{Our Contribution}
We consider a sequential decision model where each firm makes a decision on the buying quantity from its sellers, and the selling quantity and price to its buyers, given that all its sellers have made their decisions. The single source producer of the network starts the decision making with a fixed material cost. The single ending market of the network accepts all the goods offered by the firms in the last-tier and the market price is an affine decreasing function of the total quantity of goods. Each firm strives to maximize its utility, assuming that its downstream buyers make rational decisions. An \emph{equilibrium} is the collection of firm decisions such that no firm has the incentive to change its decision, assuming that its downstream buyers play rationally and the decision of other firms remain unchanged. Our main results are listed as follows.
\begin{itemize}
\item We show a linear time algorithm that finds the unique equilibrium in a series parallel network. A crucial step is to derive a closed-form expression of the price at each firm in terms of quantity.
\item We show a rich set of equilibrium comparative statics for series parallel networks, including the firm location advantage of upstream firms and a network-component-based efficiency analysis.
\item We analyze the equilibrium for generalized series parallel networks.
\end{itemize}

\subsection{Proof Techniques}
\paragraph{Equilibrium in Series Parallel Networks.} We first observe that the flow conservation property is satisfied at equilibrium, i.e. inflow equals outflow for each intermediary firm. The main strategy is to formulate the price offered to each firm in terms of its inflow. Our first algorithm starts the inductive price computation from the ending market in a reverse topological order. In the computation of each firm, we take the partial derivative of the firm's utility with respect to the buying quantity and obtain a closed-form expression for three cases. The most interesting one is the case when a single seller has multiple buyers. In this case, since the utility is a quadratic function of the buying quantity, we formulate a linear complementarity problem to compute the convex coefficients for each trade from this seller to her buyers. Once the price-quantity relation is obtained at the source, we compute the total flow needed for the network. Our second algorithm computes the flow in a topological order starting from the source producer. When a single seller encounters multiple buyers, the problem of maximizing the seller utility can be formulated by a linear complementarity problem, and this problem has an equivalent convex quadratic program. In particular, by the structure of series parallel networks, the problem can be simplified to a linear system, and the flow is distributed proportionally to the convex coefficients computed by our first algorithm. It takes linear time to find an equilibrium and the equilibrium is unique because the solution of each corresponding convex quadratic program is also unique.

\paragraph{Comparative Statics.} For the analysis of firm location advantage, we first obtain a closed-form expression of each firm's utility. By this expression, we conclude that an upstream firm which controls the flow of a downstream firm has at least twice the utility of the downstream firm. For the network-component-based efficiency analysis, we focus on analyzing the flow value and social welfare at equilibrium. We show that with the same production cost and ending market price, locally swapping the order of two components in a series composition does not change the total flow and social welfare.

\paragraph{Equilibrium Analysis for Generalized Series Parallel Networks.} We consider two extensions with multiple source producers or ending markets. When the generalized series parallel graph has a single source and multiple markets, we consider a simple network and observe that the price function of the intermediary firm can be piecewise linear and discontinuous. This enforces the source to apply either the high price or low price strategy. There can be multiple pure strategy equilibria when both strategies give the source the same utility. When the generalized series parallel graph has multiple sources and a single ending market, an equilibrium may not exist.

\subsection{Related work}
In a series parallel network, intermediary firms can be considered as Cournot markets at equilibrium. Thus, the structure of the game is closely related to the literature on Cournot games in networks. \cite{bimpikis2014cournot} and \cite{pang2017efficiency}, for example, consider a Cournot game in two sided markets. \cite{nguyen2018welfare} studies Cournot game in three-tier networks. However, the 2-tier structure of the network in these papers, and the assumption that only the middle tiers make decision in \cite{nguyen2018welfare} assumes away the complex sequential decision making considered in this work. \cite{nava2015efficiency} studies a Cournot game in general networks. However, firms are assumed to make simultaneous decisions. Simultaneous games are easier to analyze, but do not capture the essential element of sequential decision making of firms in  supply chain networks.
 
\cite{carr2005competition} considers assembly network where agents make sequential decision, but assumes a tree network. The analysis for a tree network is substantially simpler, because each firm has a single downstream node that it can sell the products to. In our game, each firm needs to make decision on the selling quantity and price to each of its buyer.  As we show, some of the quantities on some of the links can be zero. Such ``inactive'' links make the analysis more complicated. Recently, \cite{bimpikis2019supply} also considers a sequential game. The network considered in this paper is however symmetric and its structure is linear. The focus of \cite{bimpikis2019supply} is on the uncertainty of yields, which is different from the motivation in our paper.

More broadly, this work belongs to the growing literature of network games and their applications in supply chains, including \cite{corbett2001competition,federgruen2016sequential,nguyen2016delay}, and \cite{nguyen2017local}. These papers, however, are different from ours in the main focus as well as the modeling approach. 
\cite{corbett2001competition} for example, assumes a linear structure of supply chains, \cite{federgruen2016sequential} considers price competition in two-tier networks, while \cite{nguyen2016delay}and \cite{nguyen2017local} analyze  bargaining games in networks with simpler structures. The main contribution of our paper to this line of work is  a tractable analysis of sequential competition model in series parallel graphs, which allows  for a richer comparative analysis  and deeper understanding of how basic network elements influence market outcomes.

\subsection{Organization}
This work is organized as follows. In section~\ref{sec: model}, we introduce the model and the parallel serial networks together with the compositions. In section~\ref{sec: clearance}, we provide the algorithm to compute the unique equilibrium. In section~\ref{sec: eqp}, we analyze how firm location affects individual utility and how network structure influences the efficiency. In section~\ref{sec: ext}, we discuss extensions to other classes of networks and show that pure strategy equilibrium might not exist in general networks.

\section{Preliminaries} \label{sec: model}

We introduce the sequential decision mechanism and the definition of {\em series parallel graph}.

\subsection{The Model}

\paragraph{The Sequential Decision Game} Let $G=(V \cup \{s, t\}, E)$ be a simple directed acyclic network that represents an economy where  $s$ is the producer firm at the source, $t$ is the sink market, and $V$ represents intermediary firms. The arcs of $G$ represent the possibility of trade between two agents. The direction of an arc indicates the direction of trade.  The outgoing end of the arc corresponds to the seller, and  the incoming end is the buyer, while $s$ has only outgoing arcs, and $t$ has only incoming arcs. 
 The remaining vertices $i\in V$  representing intermediary  firms has both incoming and outgoing arcs. 
  For a vertex $i$, $B(i)$ and $S(i)$ are the sets of agents that can be buyers and sellers in a trade with $i$, respectively.


Agents start making their decision after the output of their upstream suppliers is determined. Each firm $i$ decides on how much to buy from each of its sellers, how much to sell to each of its buyer, and the price to sell to each of its buyer. Formally, $i$'s decision includes:
\begin{itemize}
\item The buying quantity $x^{in}_{ki} \geqslant 0$ of arc $ki \in E$ for every $k \in S(i)$.
\item The selling quantity $x^{out}_{ij} \geqslant 0$ of arc $ij \in E$ to every $j \in B(i)$.
\item The selling price $p_{ij} \geqslant 0$ of arc $ij \in E$ to every $j \in B(i)$.
\end{itemize}

The source producer does not buy any goods, so the decision of $s$ is the selling quantity $x^{out}_{sj} \geqslant 0$ and the price $p_{sj}$ of arc $sj$ to every $j \in B(s)$.

The production cost $p_s$ of the source $s$ is given and assumed to be a constant $a_s$.
\[
p_s = a_s \text{ where } a_s > 0.
\]

The sink node $t$ does not represent a firm, it corresponds to an end market. The price function $p_t$ at sink node $t$ is given and assumed to be an affine decreasing function on the total amount of goods, $X_t$, sold to the market $t$.
$$
p_{t} =  a_t - b_t X_t, \text{ where } X_t = \sum_{i \in S(t)}{x^{in}_{it}} = \sum_{i \in S(t)}{x^{out}_{it}}, a_t > a_s, \text{ and } b_t > 0.
$$
$a_t$ is the \emph{demand} of the market $t$. We note that the market must accept all the goods thus does not have a choice to reject. That is, $x^{in}_{it} = x^{out}_{it}$ for each $i \in S(t)$. Generally, for a trade  $ij \in E$, the buyer $j$ cannot obtain more than what the seller $i$ offers, thus $x^{in}_{ij} \leqslant x^{out}_{ij}$. We assume that each intermediary firm $i$ cannot get goods from any other source besides its sellers. Firms do not get any value from retaining the goods.

The payoff of the source firm $s$ is
\begin{equation} \label{eq:sp}
\Pi_s = \sum_{j \in B(s)}{p_{sj} x^{in}_{sj}} - p_s \sum_{j \in B(s)}{x^{out}_{sj}}.
\end{equation}

The utility of an intermediary firm $i \in V$ is 
\begin{equation} \label{eq:fp}
\Pi_i = \sum_{j \in B(i)}{p_{ij} x^{in}_{ij}} - \sum_{k \in S(i)} {p_{ki} x^{in}_{ki}}.
\end{equation}
The formula decomposes the utility function into two terms: the total revenue from $j \in B(i)$ and the total cost of materials from $k \in S(i)$.



The timing of the game is as follows. The source producer makes its decision first. A firm makes its decision on the buying quantity from its upstream, and the selling quantity and price to its downstream, once all of its sellers have made their decision. When deciding their accepting quantities to maximize their profits, firm $i$ also needs to take into account the strategies of both the competing firms and the firms  downstream. Firms make their decisions based on rational expectation of other firms' strategies.

\paragraph{Equilibrium Characteristics} Intuitively, an {\em equilibrium} of the sequential decision game is an assignment of good quantity and price such that no firm is willing to change its decision after knowing the decision of other non-downstream firms, and assuming that all the downstream firms also pick their best decisions.

We present two examples to illustrate the equilibrium concept. The first example is a line network. The source producer controls the amount and price, thus affects the decision of the intermediary firm.

\begin{example}[The Equilibrium of a Line Network] \label{ex:eq-ln}
\mbox{}\\
Consider the following line network.

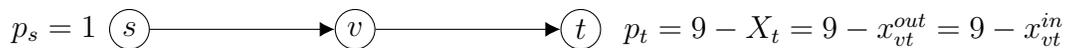
\begin{center}
\begin{tikzpicture}
	\node[vertex] (s) at (0,0) {$s$}; 
	\node at (-1, 0) {$p_{s} = 1$};
	\node[vertex] (a) at (3,0) {$v$}; 
    \node[vertex] (t) at (6,0) {$t$};
	\node at (9.5, 0) {$p_{t} = 9 - X_t = 9 - x^{out}_{vt} = 9 - x^{in}_{vt}$};
	\path[->]
		(s) edge (a) 
		(a) edge (t);
\end{tikzpicture}
\end{center}

Suppose $s$ offers $v$ $x^{out}_{sv} = 2$ and $p_{sv} = 7$, $v$ has to make a decision on $x^{in}_{sv}$ and $x^{out}_{vt}$. $v$ cannot decide $p_{vt}$ because it is already fixed as $9 - x^{out}_{vt}$. The sink market $t$ accepts all the goods, so the utility of $v$ is
\[\Pi_v = (9 - x^{out}_{vt}) x^{in}_{vt} - p_{sv} x^{in}_{sv} = (9 - x^{out}_{vt}) x^{out}_{vt} - 7 x^{in}_{sv}.\]

If $v$ is a rational player who maximizes $\Pi_v$, then $v$ will sell all the goods it bought, so $x^{in}_{sv} = x^{out}_{vt}$ and $\Pi_v = (9 - x^{in}_{sv}) x^{in}_{sv} - 7 x^{in}_{sv}$. By taking the derivative:
\[\frac{d \Pi_v}{d x^{in}_{sv}} = 9 - 2 x^{in}_{sv} - 7 = 0 \implies x^{in}_{sv} = 1,\]
so $b$ will accept $1$ out of $2$ units of the goods from $s$.

The utility of $s$ is
\[\Pi_s = p_{sv} x^{in}_{sv} - p_s x^{out}_{sv} = 7 \times 1 - 1 \times 2 = 5.\]

In fact, $s$ is over-selling the goods, it could have set $x^{out}_{sv} = 1$ instead such that
\[\Pi_s = p_{sv} x^{in}_{sv} - p_s x^{out}_{sv} = 7 \times 1 - 1 \times 1 = 6 > 5.\]

However, this is not the best strategy for $s$. If $s$ offers $v$ $x^{out}_{sv} = 2$ and $p_{sv} = 5$, by the same reasoning, the utility of $v$ is
\[\Pi_v = (9 - x^{in}_{sv}) x^{in}_{sv} - 5 x^{in}_{sv}.\]

$v$ tries to maximize $\Pi_v$:
\[\frac{d \Pi_v}{d x^{in}_{sv}} = 9 - 2 x^{in}_{sv} - 5 = 0 \implies x^{in}_{sv} = 2,\]
so $v$ will accept all the goods from $s$.

This time, the utility of $s$ is
\[\Pi_s = p_{sv} x^{in}_{sv} - p_s x^{out}_{sv} = 5 \times 2 - 1 \times 2 = 8,\]
and $s$ is better off. This is the optimal strategy for $s$. In summary, the equilibrium is
\begin{center}
    \begin{tabular}{ | l | l | l | l |}
    \hline
    $p_{sv}$ & $x^{out}_{sv}$ & $x^{in}_{sv}$ & $x^{out}_{vt}$ \\ \hline
		5 & 2 & 2 & 2 \\ \hline		
    \end{tabular}
\end{center}
\end{example}

The second example illustrates the competition between two intermediary firms. The source producer controls the quantity and price, thus affects the decisions of its two downstream firms.

\begin{example}[The Equilibrium with Two Intermediary Firms] \label{ex:eq-simple-spg}
\mbox{}\\
Consider the following network.
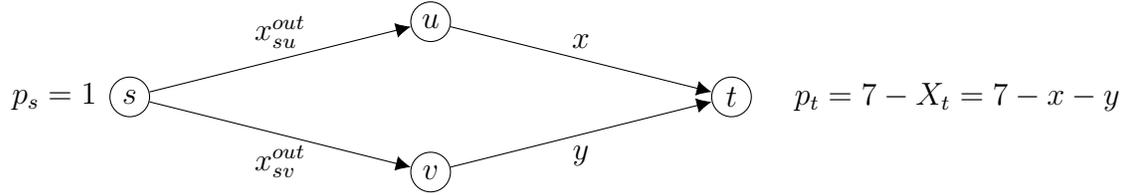
\begin{center}
\begin{tikzpicture}[baseline=0,scale=2]
	\node at (-0.5, 0) {$p_{s} = 1$};
	\node[vertex] (s) at (0,0) {$s$}; 
	\node[vertex] (b) at (2,0.5) {$u$}; 
	\node[vertex] (c) at (2,-0.5) {$v$};
	\node[vertex] (d) at (4,0) {$t$};
	\node at (5.5,0) {$p_t = 7 - X_t = 7 - x - y$};
	\path[->]
		(s) edge node[above] {$x^{out}_{su}$} (b)
		(s) edge node[below] {$x^{out}_{sv}$} (c)
		(b) edge node[above] {$x$} (d)
		(c) edge node[below] {$y$} (d)
;
\end{tikzpicture}
\end{center}

Suppose the decision of $s$ is
\begin{center}
    \begin{tabular}{ | l | l | l | l |}
    \hline
    $p_{su}$ & $p_{sv}$ & $x^{out}_{su}$ & $x^{out}_{sv}$ \\ \hline
		3 & 4 & 1 & 1 \\ \hline		
    \end{tabular}
\end{center}
and $u$ and $v$ are rational firms, i.e. they will sell all the goods they bought to maximize their payoff. For simplicity, let $x=x^{in}_{su}$ and $y=x^{in}_{sv}$. The utilities of $u$ and $v$ are
\[\Pi_u = (7 - x - y)x - 3 x \quad \text{and} \quad \Pi_v = (7 - x - y)y - 4 y.\]
By taking the derivative, we have
\[\frac{d \Pi_u}{dx} = 4 - 2x - y \quad \text{and} \quad \frac{d \Pi_v}{dy} = 3 - x - 2y.\]

We claim that the best decision of $u$ and $v$ is $x=1$ and $y=1$. $\Pi_v$ is a concave function so it is maximized when $\frac{d \Pi_v}{dy}=0$, which implies the best response to $x=1$ is $y=1$. We observe that $\Pi_u$ is also concave. $\Pi_u$ is maximized when $\frac{d \Pi_u}{dx} = 0$. However, when $y=1$, $x$ cannot be $\frac{3}{2}$ since $x \leqslant x^{out}_{su} = 1$. While $y=1$ and $x\in[0,1]$, $\Pi_u$ is an increasing function, so the best response to $y=1$ is $x=1$.

The utility of $s$ for this decision is
\[\Pi_s = p_{su} x + p_{sv} y - p_s (x^{out}_{su} + x^{out}_{sv}) = 3 \times 1 + 4 \times 1 - 1 \times (1+1) = 5.\]

Given that $p_{su}=3$ and $p_{sv}=4$, there is a better quantity decision for $s$. We recall that $\Pi_u$ and $\Pi_v$ are both concave, so it suffices to show that $\frac{d \Pi_u}{dx}$ and $\frac{d \Pi_v}{dx}$ are both zeros. This happens when $x=\frac{5}{3}$ and $y=\frac{2}{3}$. If $x^{out}_{su}=\frac{5}{3}$ and $x^{out}_{sv}=\frac{2}{3}$, then $u$ and $v$ will buy all the goods from $s$ and the utility of $s$ is
\[\Pi_s = p_{su} x + p_{sv} y - p_s (x^{out}_{su} + x^{out}_{sv}) = 3 \times \frac{5}{3} + 4 \times \frac{2}{3} - 1 \times (\frac{5}{3}+\frac{2}{3}) = \frac{16}{3}.\]

However, this is not the best decision of $s$. The equilibrium for this network is as follows.
\begin{center}
    \begin{tabular}{ | l | l | l | l | l | l |}
    \hline
    $p_{su}$ & $p_{sv}$ & $x^{out}_{su}$ & $x^{out}_{sv}$ & $x$ & $y$ \\ \hline
		4 & 4 & 1 & 1 & 1 & 1 \\ \hline		
    \end{tabular}
\end{center}
One can verify that $\Pi_v$ and $\Pi_v$ are concave and the derivatives are zeros. The payoff of $s$ is
\[\Pi_s = p_{su} x + p_{sv} y - p_s (x^{out}_{su} + x^{out}_{sv}) = 4 \times 1 + 4 \times 1 - 1 \times (1+1) = 6.\]
\end{example}

We observe that the best strategy for each firm $i \in V$ is to always sell as much as bought since it cannot benefit from paying more for those unsold goods. At the selling side, suppose firm $i$ is willing to offer ${x^{out}_{ij}}$ quantity of goods to firm $j$, but part of the goods got rejected, i.e. $x^{in}_{ij} < x^{out}_{ij}$. This can never happen in an equilibrium, because $i$ will be better off by rejecting $x^{out}_{ij} - x^{in}_{ij}$ amount of goods from its upstream before selling.

The next observation lists the properties of supplying quantities at an equilibrium.

\begin{observation}[Equilibrium Flow Conservation] \label{obs: fc}
An equilibrium satisfies:
\begin{enumerate}
\item $x^{out}_{ij} = x^{in}_{ij}$ for each $ij \in E$.
\item $\sum_{k \in S(i)}{x^{in}_{ki}} = \sum_{j \in B(i)}{x^{out}_{ij}}$, i.e. inflow is equal to outflow for each firm $i \in V$.
\end{enumerate}
\end{observation}

Now we are ready to state the formal definition of an equilibrium.

\begin{definition}
An \emph{equilibrium} is a set of strategies including:
\begin{enumerate}
\item the strategy of the source producer $s$: $p^{out}_{sj}$ and $x^{out}_{sj}$ for $j \in B(s)$, and
\item the strategy of each intermediary firm $i \in V$: $p^{out}_{ij}$ (if $j \neq t$, otherwise $p^{out}_{ij} = p_{t} =  a_t - b_t X_t$), $x^{out}_{ij}$ for $j \in B(s)$, and $x^{in}_{ki}$ for $k \in S(i)$
\end{enumerate}
such that
\begin{enumerate}
\item $x^{in}_{ij} = x^{out}_{ij}$ for each $ij \in E$, i.e. $j$ accepts all the goods $i$ offers, and
\item for each firm $i \in \{s\} \cup V$, $i$ does not have the incentive to change its strategy for a better payoff, assuming that each descendant firm of $i$ plays the best strategy that maximizes its payoff, and the strategy of non-descendant firms of $i$ remain the same.
\end{enumerate}
\end{definition}

For later notations, at an equilibrium, we denote $x_{ij}$ as the \emph{flow} of arc $ij$, i.e. $x_{ij} = x^{out}_{ij} = x^{in}_{ij}$, and no longer use $x^{in}_{ij}$ and $x^{out}_{ij}$. Meanwhile, since each firm accepts all the offers and sells everything they bought, we denote this sum of flow as $X_i := \sum_{k \in S(i)}{x_{ki}} = \sum_{j \in B(i)}{x_{ij}}$. The utility of firm $i$ in \eqref{eq:fp} becomes
\begin{equation} \label{eq:fp2}
\Pi_i = \sum_{j \in B(i)}{p_{ij} x_{ij}} - \sum_{k \in S(i)}{p_{ki}x_{ki}}
\end{equation}
and the utility of source firm $s$ in \eqref{eq:sp} becomes
\begin{equation} \label{eq:sp2}
\Pi_s = \sum_{j \in B(s)}{p_{sj} x_{sj}} - p_s \sum_{j \in B(s)}{x_{sj}}.
\end{equation}

For flow activities, an arc $ij \in E$ is {\em active} if $x_{ij} > 0$, and {\em inactive} if $x_{ij} = 0$. For each firm $i \in V$ and active arcs $ki \in E$ and $ij \in E$, $p_{ki} \leqslant p_{ij}$. That is, the buying price should not exceed the selling price. Otherwise, $i$ could have been better off by rejecting some goods from $k$ and choosing not to offer the same amount of goods to $j$.

\begin{observation} \label{obs:inc-p}
For every $ki \in E$ and $ij \in E$ that are active, the price at an equilibrium satisfies $p_{ki} \leqslant p_{ij}$.
\end{observation}

To define equilibrium uniqueness, we require the set of active arcs to be unique, as well as the flow and price of each active arc. The prices of inactive trades, on the other hand, can be arbitrary since they do not contribute to the seller revenue.

\begin{definition}
An equilibrium of a network $G$ is \emph{unique} if the set of active arcs is unique, as well as the flow and price of each active arc.
\end{definition}

\subsection{Series Parallel Graphs}

\paragraph{General Series Parallel Graphs} We consider the case when $G$ is a {\em series parallel graph} (SPG). The networks in Example~\ref{ex:eq-ln} and \ref{ex:eq-simple-spg} are both SPGs. Our main goal is to compute the equilibrium in networks that belong to this special graph family. This class of networks is well studied and has several applications in graph theory (see for example \cite{duffin1965topology}). For completeness, we  provide a formal definition as  follows. 
\begin{definition}[SPG] \label{def: 2.4}
A {\em single-source-and-sink SPG} is a graph that can be constructed by a sequence of series and parallel compositions starting from a set of copies of a single-arc graph, where:
\begin{enumerate}
\item {\em Series composition} of $X$ and $Y$: given two SPGs $X$ with source $s_X$ and sink $t_X$, and $Y$ with source $s_Y$ and sink $t_Y$, form a new graph $G=S(X,Y)$ by identifying $s=s_X$, $t_X=s_Y$, and $t=t_Y$.
\item {\em Parallel composition} of $X$ and $Y$: given two SPGs $X$ with source $s_X$ and sink $t_X$, and $Y$ with source $s_Y$ and sink $t_Y$, form a new graph $G=P(X,Y)$ by identifying $s=s_X=s_Y$ and $t=t_X=t_Y$. 
\end{enumerate}
\end{definition}

\paragraph{Shortcut-free Series Parallel Graphs}

We start with the definition of \emph{shortcuts}.

\begin{definition}
    Given an SPG $G=(V,E)$, let $i,j \in V$. Consider a path $l_{ij}=(i, v_1, ..., v_k, j)$ from node $i$ to node $j$. If there is an arc $ij \in E$, then we say $ij$ is a {\em shortcut} of $l_{ij}$, or $ij$ {\em dominates} path $l_{ij}$.
\end{definition}

\begin{definition}
An SPG is \emph{shortcut-free} if it has no shortcuts, i.e. there is no path dominated by an arc.    
\end{definition}

Node $k$ is a {\em parent node} of $i$ if there is a directed path from $k$ to $i$. The set of parent nodes of $i$ is denoted as $P(i)$. Similarly, we can define child nodes and the set of children $C(i)$. Given a shortcut-free SPG, if we consider the relation between direct (or adjacent) parent and child $i$ and $j$, where $ij \in E$, there are three possibilities\footnote{Multiple sellers and multiple buyers case does not exist in shortcut-free SPGs, we refer the proof to Appendix~\ref{app:no_mm}}:

\begin{itemize}
	\item Single seller and single buyer, $|S(j)| = |B(i)| = 1$. (SS)
	\item Multiple sellers and single buyer, $|S(j)| \geqslant 2, |B(i)| = 1$. (MS)
	\item Single seller and multiple buyers, $|S(j)| = 1, |B(i)| \geqslant 2$. (SM)
\end{itemize}

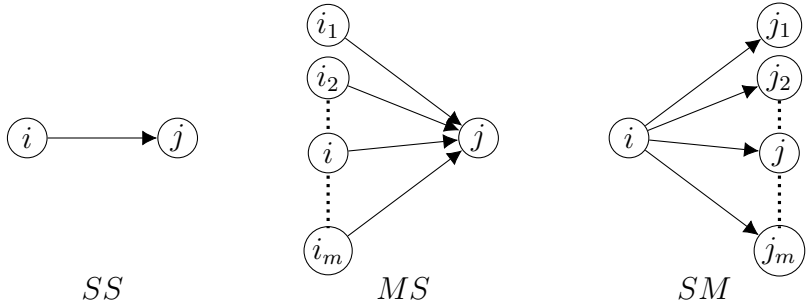
\begin{center}
\begin{tikzpicture}
	\node[vertex] (i) at (-2,0) {$i$}; 
	\node[vertex] (j) at (0,0) {$j$}; 
	\path[->]
		(i) edge (j);
	\node at (-1, -2) {$SS$};

	\node[vertex] (i1) at (2,1.5) {$i_1$}; 
	\node[vertex] (i2) at (2,0.8) {$i_2$};
	\node[vertex] (ik) at (2,-1.5) {$i_m$}; 
	\node[vertex] (i) at (2,-0.2) {$i$};
	\node[vertex] (j) at (4,0) {$j$}; 
	\draw[dotted, very thick] (2,0.5) -- (2,0.1);
	\draw[dotted, very thick] (2,-0.5) -- (2,-1.2);
	\path[->]
		(i) edge (j)
		(i1) edge (j)
		(i2) edge (j)
		(ik) edge (j);
	\node at (3, -2) {$MS$};

	\node[vertex] (j1) at (8,1.5) {$j_1$}; 
	\node[vertex] (j2) at (8,0.8) {$j_2$};
	\node[vertex] (jk) at (8,-1.5) {$j_m$}; 
	\node[vertex] (i) at (6,0) {$i$};
	\node[vertex] (j) at (8,-0.2) {$j$}; 
	\draw[dotted, very thick] (8,0.5) -- (8,0.1);
	\draw[dotted, very thick] (8,-0.5) -- (8,-1.2);
	\path[->]
		(i) edge (j1)
		(i) edge (j2)
		(i) edge (j)
 		(i) edge (jk);
	\node at (7, -2) {$SM$};
\end{tikzpicture}
\end{center}


Sometimes there are multiple paths from a parent node to one of its children, we call these paths {\em disjoint} if they do not have any common intermediary nodes, that is, all nodes except the starting and the ending ones are different. Base on this definition, we can define the merging nodes with respect to node $i$.


\begin{definition}[Self-merging Child Node] \label{def: 3.4}
	Node $j \in C(i)$ is a {\em self-merging child node} of $i$ if there are disjoint paths from $i$ to $j$. The set of such nodes $j$ is termed $C_S(i)$.
\end{definition}

\begin{definition}[Parent-merging Child Node] \label{def: 3.3}
	Node $j \in C(i)$ is a {\em parent-merging child node} of $i$, if there exists a node $k \in P(i)$, such that there are disjoint paths from $k$ to $j$. The set of such nodes $j$ is denoted as $C_P(i)$.
\end{definition}

For $ij \in E$, we introduce the set of special self-merging child nodes of $i$ and its direct child $j$ as $C_T(i,j) := C_S(i) \cap C(j) \backslash C_P(i)$. This notation helps us capture the ``internal'' merging nodes that are responsible for the selling price and quantity offered to $i$.

\begin{observation} \label{obs:mc}
An SPG has the following properties:
\begin{enumerate}
	\item $C_P(s) = C_P(t) = \emptyset$.
	\item In the SS case for $ij \in E$, $C_P(j) = C_P(i)$.
	\item In the SM case for $ij \in E$, $C_P(j) = C_P(i) \sqcup C_T(i,j)$.
	\item In the MS case for $ij \in E$, $C_P(i) = C_P(j) \sqcup \{j\}$.
\end{enumerate}
\end{observation}

We note that $\sqcup$ stands for disjoint set union. We show Example~\ref{app:ex:mc} in the appendix for merging child node sets.

\section{Equilibrium Computation} \label{sec: clearance}

We present two algorithms to compute the equilibrium quantity and price for each arc. We start with shortcut-free SPGs and show that all arcs are active at equilibrium. To do so, we first derive a closed-form relation between the quantity and price offered to the firms at equilibrium via a backward algorithm. Then, we show that the unique optimal quantity and price offered to each firm can be solved following the decision sequence from the source to sink by the closed-form relation. For SPGs with shortcuts, we show that the trade on paths dominated by shortcuts are inactive. Thus, the equilibrium for an SPG with shortcuts can be found by the same algorithm after removing dominated paths.

\subsection{Shortcut-free Series Parallel Graphs}

\paragraph{Equilibrium Price Computation.} A key characteristic of the equilibrium is that all edges are active when there are no shortcuts. The equilibrium price has a closed-form expression in terms of the equilibrium quantity based on the structure of the SPG.

\begin{restatable}{theorem}{lemeqprice}\label{lem:pq-relation}
Given a shortcut-free SPG $G$, at an equilibrium, all arcs are active. For each firm $i \in V$, each seller $k \in S(i)$ offers $i$ the same price $p_i$, and the following relation holds
\[
p_i = a_t - b_i X_i - \sum_{l \in C_P(i)} b_l X_l
\]
where $b_i$ for each $i \in V \cup \{s\}$ is a positive constant that only depends on the structure of $G$.
\end{restatable}

Theorem~\ref{lem:pq-relation} shows a concise way to present the price and quantity relation at equilibrium. The high level strategy for deriving the closed-form expression is via a backward induction starting from the sink market $t$. In order to calculate the subgame perfect equilibrium, we consider the fact that the upstream firms make their decision based on the best decision of the downstream. Since the price of the sink market $t$ is an affine decreasing function, we can inductively show that the utility of each intermediary firm is a concave function of the quantity. We observe that the derivative of the utility with respect to the quantity of a trade cannot be positive since otherwise the upstream firm can be better off by raising the price. On the other hand, the derivative with respect to the quantity of an active trade must be zero. This is because the upstream firms assume that the downstream firms make their decision to maximize their utility, which implies the derivative is zero. By this observation, we derive a linear complementarity problem, and show that the best price offered to the downstream is an affine decreasing function of the quantity. The proof is given in Appendix~\ref{app:lem:pq-relation}. By adapting the main equations in that proof, we introduce Algorithm~\ref{alg: 1} to compute the price at equilibrium in terms of the quantity.

Algorithm~\ref{alg: 1} starts from the sink market $t$. In each iteration, given the downstream $b_j$ where $j \in B(i)$, we compute $b_i$ and this can be done in $O(deg^+(i))$ time where $deg^+(i)$ is the out-degree of $i$. The calculation is in a reverse topological order. Particularly, in the SM case, we store the {\em convex coefficients} of each downstream node $j \in B(i)$ which is used later in the equilibrium quantity computation. The number of the $b_i$ computation is bounded by $O(|V|)$. Therefore, it takes linear time to compute the price function in terms of quantity by Algorithm~\ref{alg: 1}. When $s$ is reached, we already have $p_s = a_t - b_s X_s = a_s$ since $C_P(s) = \emptyset$. Eventually, $X_s = \frac{a_t - a_s}{b_s}$ so the expected price of $s$ meets the given production cost. We show Example~\ref{app: ex: price_compute_alg1} in the appendix for the price calculation by Algorithm~\ref{alg: 1}.

\begin{algorithm}[H] 
\caption{: Price Function Computation (Backward)}
\label{alg: 1}
\renewcommand{\algorithmicrequire}{\textbf{Input:}}
\renewcommand{\algorithmicensure}{\textbf{Output:}}

\algorithmicrequire{ A shortcut-free SPG $G = (V \cup \{s,t\}, E)$, the price function $p_t = a_t - b_t X_t$, and the source production cost $p_s = a_s$.}

\algorithmicensure{ The equilibrium price function $p_i$ for each $i \in V$, the convex coefficients $\alpha_j$ for arc $ij$ where $j \in B(i)$ in the SM case, and the source flow $X_s$.}
\begin{algorithmic}[1]
\State Starting from $t$, given the downstream buyer(s) $j$'s price function $p_j$, compute the upstream seller(s) $i$'s price case by case in a reverse topological order:

\begin{itemize}
	\item For the SS case,
	\begin{align} \label{eq: SS} \tag{SS}
	b_i = 2 b_j + \sum_{l \in C_P(j)} b_l.
	\end{align}
	\item For the MS case, for each seller $i$,
	\begin{align} \label{eq: MS} \tag{MS}
	b_i = b_j + \sum_{l \in C_P(j)} b_l.
	\end{align}
	\item For the simple SM case ($|C_S(i)| = 1$)\footnotemark, suppose $b_{j}$ was calculated for all $j \in B(i)$,
	\begin{align*} \label{eq: SM} \tag{Simple SM}
	b_i = \frac{2}{\sum_{j \in B(i)} \frac{1}{b_{j}}} + 2 \sum_{l \in C_S(i) \setminus C_P(i)} b_l + \sum_{l \in C_P(i)} b_l.
	\end{align*}
	For each $j \in B(i)$, assign the convex coefficient $\alpha_j = \frac{\frac{1}{b_j}}{\sum_{j' \in B(i)} \frac{1}{b_{j'}}}$ to arc $i j$.
\end{itemize}

\State Set the price function at seller $i$: $p_i = a_t - b_i X_i - \sum_{l \in C_P(i)} b_l X_l$.
\If {seller $i$ is the source $s$}
	\State {Set $X_s = \frac{a_t - a_s}{b_s}$.}
	\State \Return
\EndIf

\end{algorithmic}
\end{algorithm}

\footnotetext{If $|C_S(i)| \geqslant 2$, the computation of $b_i$ is more complicated, the detail is provided in Appendix~\ref{para:sm}. We also show Example~\ref{app: ex: price_compute} for this case in the appendix.}

\paragraph{Equilibrium Quantity Computation.} \label{subsubsec:equi_quant_comp}
We consider shortcut-free SPGs. After having the closed-form relation between the equilibrium price and quantity, we present an algorithm that finds the unique equilibrium. Consider the quantity decision for firm $i$ to its downstream buyers $j \in B(i)$. Suppose firm $i$ only has one buyer, i.e., $|B(i)| = 1$, by Observation~\ref{obs: fc}, inflow equals outflow at firm $i$, and firm $j$ will accept all the offer from $i$, formally, $x_{ij} = X_i$. Hence, in the following analysis, we focus on the nontrivial case, the SM case, when firm has multiple downstream buyers, i.e., $|B(i)| \geqslant 2$. How to assign the goods to different buyers so that the utility of firm $i$ is maximized? We recall that by Theorem~\ref{lem:pq-relation}, all arcs are active and for each firm $j \in V$, $p_{ij}=p_j$ for each $i \in S(j)$. Therefore, the utility of $i$ in \eqref{eq:fp2} becomes
\begin{align} \label{eq:fp3}
\Pi_i = \sum_{j \in B(i)} p_j x_{ij} - p_i \sum_{j \in B(i)} x_{ij}.
\end{align}
Again by Theorem~\ref{lem:pq-relation}, the price function of seller $j$ is
\begin{align} \label{eq: q2}
p_j = a_t - b_j x_{ij} - \sum_{l \in C_P(j)} b_l X_l
\end{align}
since $X_j = x_{ij}$ in the SM case.

The utility of firm $i$ in \eqref{eq:fp3} is concave. At the equilibrium, if $x_{ij} > 0$, then $\frac{\partial \Pi_i}{\partial x_{ij}} = 0$; if $x_{ij} = 0$, then $\frac{\partial \Pi_i}{\partial x_{ij}} \leqslant 0$. Therefore, solving the optimal decision for firm $i$ is equivalent to solving the following linear complementarity problem (LCP) with variables $x_{ij}$ where $j \in B(i)$.
\begin{equation} \label{eq: lcp} \tag{LCP}
\begin{cases}
\sum_{j \in B(i)} \frac{\partial \Pi_i}{\partial x_{ij}} x_{ij}  = 0,\\
\frac{\partial \Pi_i}{\partial x_{ij}}  \leqslant 0  \quad \forall j \in B(i),\\
x_{ij} \geqslant 0 \quad \forall j \in B(i).
\end{cases}
\end{equation}

To solve this feasibility problem \ref{eq: lcp} and find the optimal allocation to the downstream firm $j$, we take the derivative of $\Pi_i$ with respect to $x_{ij}$ and obtain
\begin{equation} \label{eq: q3}
\begin{aligned} 
\frac{\partial \Pi_i}{\partial x_{ij}} &= p_j + \sum_{h \in B(i)} \frac{\partial p_h}{\partial x_{ij}} x_{ih} - p_i.
\end{aligned}
\end{equation}

The second term of \eqref{eq: q3} can be expanded as
\begin{equation} \label{eq: q4}
\begin{aligned} 
\sum_{h \in B(i)} \frac{\partial p_h}{\partial x_{ij}} x_{ih}
&= - b_j x_{ij} - \sum_{h \in B(i)} (\frac{\partial \sum_{l \in C_P(h)} b_l X_l}{\partial x_{ij}}) x_{ih} \\
&= - b_j x_{ij} - \sum_{h \in B(i)} (\sum_{l \in C_P(h) \cap C(j)} b_l) x_{ih} \\
&= - b_j x_{ij} - \sum_{l \in C_T(i,j)} b_l X_l - \sum_{l \in C_P(i)} b_l X_i.
\end{aligned}
\end{equation}

The second equality holds because $X_l$ includes $x_{ij}$ only when $l$ is a child of $j$. The third equality holds by rearranging and summing the inflow value of the merging nodes. Plug \eqref{eq: q2} and \eqref{eq: q4} back into \eqref{eq: q3}, we get
\begin{equation} \label{eq: q9}
\begin{aligned} 
\frac{\partial \Pi_i}{\partial x_{ij}} 
&= a_t - 2 b_j x_{ij} - \sum_{l \in C_P(j)} b_l X_l - 
\sum_{l \in C_T(i,j)} b_l X_l - \sum_{l \in C_P(i)} b_l X_i - c_i X_i - p_i \\
&= a_t - 2 b_j x_{ij} - 
2 \sum_{l \in C_T(i,j)} b_l X_l - const_i
\end{aligned}
\end{equation}
where
\[const_i := (\sum_{l \in C_P(i)} b_l + c_i) X_i + \sum_{l \in C_P(i)} b_l X_l + p_i\]
is defined in terms of $X_i$, $X_l$ where $l \in C_P(i)$, and $p_i$. These values were determined by the upstream buyers thus are regarded as constants to $i$. The second equality of \eqref{eq: q9} holds by Observation~\ref{obs:mc}: $C_P(j) = C_P(i) \sqcup C_T(i,j)$.

We introduce a convex quadratic program (CQP) to solve \ref{eq: lcp}:
\begin{equation} \label{eq: cp} \tag{CQP}
\begin{aligned}
& \minimize_{x, X} & & \sum_{j \in B(i)} b_j x_{ij}^2 + \sum_{l \in C_S(i) \backslash C_P(i)} b_l X_l^2 \\
& \text{subject to}
& & a_t - 2 b_j x_{ij} - \sum_{l \in C_T(i,j)} 2 b_l X_l  \leqslant const_i & \forall j \in B(i), \\
& & & x_{ij} \geqslant 0 & \forall j \in B(i).
\end{aligned}
\end{equation}


By examining the KKT conditions of the quadratic program, each variable $X_l$ satisfies $X_l = \sum_{j \mid l \in C(j)} x_{ij}$ where $j \in B(i)$, which fits the definition of $X_l$. The feasibility of \ref{eq: lcp} also holds. The proof of Lemma~\ref{lem: 3.2} is provided in Appendix~\ref{app: lem: 3.2}.

\begin{restatable}{lemma}{lemlcp}\label{lem: 3.2}
	Problem~\ref{eq: lcp} is equivalent to the convex optimization problem~\ref{eq: cp}, and the solution is unique.
\end{restatable}

After the equilibrium price and quantity relation function is computed by Algorithm~\ref{alg: 1}, by solving \ref{eq: cp} directly, we have the optimal decision of each firm in polynomial time. By Theorem~\ref{lem:pq-relation}, $\frac{\partial \Pi_i}{\partial x_{ij}}=0$ since all the arcs are active. This is equivalent to solving a linear system. The algorithm can be sped up further by distributing the flow from $i$ to $j \in B(i)$ proportionally to the convex coefficients pre-computed in Algorithm~\ref{alg: 1}. Besides, each $p_j$ has the same value so that $i$ has no preference over whom to sell to. We refer the proof of Lemma~\ref{lem: 3.3} to Appendix~\ref{para:fccp}\footnote{Lemma~\ref{lem: 3.3} is used to find the price and the convex coefficients for the SM case in the proof of Theorem~\ref{lem:pq-relation}.} (the SM case).

\begin{lemma} \label{lem: 3.3}
	For the SM case, $\Pi_i$ is maximized by distributing the flow to $j \in B(i)$ proportionally to the convex coefficients pre-computed in Algorithm~\ref{alg: 1}. Besides, each $p_j$ has the same value.
\end{lemma}


By Lemma~\ref{lem: 3.3}, we introduce Algorithm~\ref{alg: 2} that takes linear time to find the equilibrium price and quantity. The algorithm starts from the source $s$ and distributes the flow to the downstream firms in a topological order. We show Example~\ref{app: ex: flow_compute_alg2} in the appendix for an equilibrium flow computation by Algorithm~\ref{alg: 2}.

\begin{algorithm}[H] 
\caption{: SPG Flow and Price Computation (Forward)}
\label{alg: 2}
\renewcommand{\algorithmicrequire}{\textbf{Input:}}
\renewcommand{\algorithmicensure}{\textbf{Output:}}

\algorithmicrequire{ A shortcut-free SPG $G = (V \cup \{s,t\}, E)$, the market price function $p_t = a_t - b_t X_t$, and the source production cost $p_s = a_s$.}

\algorithmicensure{ The equilibrium flow assignment $x_{ij}$ for $ij \in E$ and price $p_i$ for $i \in V$.}
\begin{algorithmic}[1] 
\State Get $X_s$, the price function $p_i$ for each $i \in V$, and the convex coefficients in the SM case by running Algorithm~\ref{alg: 1}.

\State Assign quantity and price according to a topological order as the following:

\begin{itemize}
	\item For the single buyer case, the flow to the buyer is the sum of the upstream flow.

	\item For the single seller and multiple buyers case, assign the downstream flow proportionally to the convex coefficients.

	\item Set the price accordingly to the quantity by the price function for each case.
\end{itemize}
\If {the buyer is sink $t$}
	\State \Return
\EndIf

\end{algorithmic}
\end{algorithm}

\subsection{General Series Parallel Graphs} Suppose the given SPG $G=(V,E)$ has a shortcut $ij \in E$ that dominates a path $l_{ij} = (i, v_1, ..., v_k, j)$. When the price $p_j$ is a decreasing function of $X_j$, $i$ always prefers selling to $j$ directly than through the intermediary firms along the path $l_{ij}$ in order to obtain better utility. We refer the proof details to Appendix~\ref{app:lem:shortcut} and show Example~\ref{app:ex:shortcut} in the appendix that illustrates an equilibrium with inactive trades in an SPG with a shortcut.

\begin{restatable}{lemma}{lemshortcut}\label{lem:shortcut}
Given an SPG, at an equilibrium, if $ij \in E$ is a shortcut of a path $l_{ij}$ and the price $p_j$ is a decreasing function of $X_j$, then there is no trade on $l_{ij}$, i.e. all the arcs on the path $l_{ij}$ are inactive.
\end{restatable}

In the price computation for shortcut-free SPGs, we show by induction that $p_j$ is an affine decreasing function of $X_j$. By Lemma~\ref{lem:shortcut}, given a general SPG, we can remove the dominated paths and obtain a shortcut-free SPG. The equilibrium quantity and price can be found by Algorithm~\ref{alg: 1} and Algorithm~\ref{alg: 2} in linear time. The uniqueness of equilibrium can be proved by encoding this problem into \ref{eq: lcp} and its corresponding \ref{eq: cp} has a unique solution. We conclude by the following theorem.

\begin{theorem} \label{thm: 3.1}
	There exists a linear time algorithm to solve the equilibrium quantity and price for SPGs, and the equilibrium is unique.
\end{theorem}

\section{Structural Analysis} \label{sec: eqp}

We present some structural analyses, including the relation between firm location and utility, and the influence and invariants of different SPG component compositions on the equilibrium. We consider shortcut-free SPGs throughout this section.

\subsection{Firm Location and Individual Utility}
This section focuses on firm's utility at equilibrium. Specifically, how does the position of a firm in the network influence its utility at equilibrium? The following example is useful to address this question.

\begin{example}[Firm Utility in a Line Network]
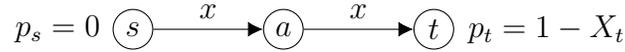
\mbox{}\\
\begin{center}
\begin{tikzpicture}
	\node[vertex] (s) at (-2, 0) {$s$}; 
	\node at (-3, 0) {$p_s = 0$};
	\node[vertex] (a) at (0, 0) {$a$}; 
	\node[vertex] (t) at (2, 0) {$t$}; 
	\node at (3.5, 0) {$p_{t} = 1 - X_t$};
	\path[->]
		(s) edge node [above] {$x$} (a)
		(a) edge node [above] {$x$} (t);
\end{tikzpicture}
\end{center}

The price at firm $a$ is $p_a = 1 - 2x$ and at producer $s$ is $p_s = 1 - 4x$. Therefore, the utility at firm $a$ is $\Pi_a = (p_t - p_a)x = x^2$ and at producer $s$ is $\Pi_s = 2x^2 = 2\Pi_a$.
\end{example}

The example above shows an intuition of the location advantage, that the firm closer to the source may have higher utility than its downstream buyers. However, this is not always true, especially when there are strong competition among upstream buyers (i.e. the MS case). The upstream firm who controls all the flow of its downstream firm has a relatively better utility at equilibrium. Therefore, we introduce the following definition.
\begin{definition}[Dominating Parent]
	$i$ is a \emph{dominating} parent of $j$ if all the path from source $s$ to $j$ must go through $i$.
\end{definition}

Before analyzing the utility relation between a dominating parent and a dominated child, let us first focus on individual utility. By using the coefficient relation between the buyer $i$ and the seller $j \in B(i)$ as in equation~\ref{eq: SS}, \ref{eq: MS}, and \ref{eq: SM}, we show the closed-form expression of the utility in Lemma~\ref{thm: 4.4}. The proof is provided in Appendix~\ref{app: lem: 4.3}.
\begin{restatable}{lemma}{lemucf}\label{lem: 4.3}
\[\Pi_i = \frac{1}{2} (b_i + \sum_{l \in C_P(i)} b_l) X_i^2 \quad \forall i \in V \cup \{s\}.\]
\end{restatable}

From Lemma~\ref{lem: 4.3}, we show a closed-form expression of the price offered by the source, which is irrelevant to the structure of the supply chain network. Then, we show the \emph{double utility rule} of a dominating parent.

The utility of the source is
\[\Pi_s = \frac{1}{2}b_s X^2_s = \frac{1}{2}b_sX_s \frac{a_t-a_s}{b_s} = \frac{a_t-a_s}{2}X_s.\]
By Lemma~\ref{lem: 3.3}, $s$ offers its buyers the same price at equilibrium. Let $p = p_j$ for $j \in B(s)$, we have
\[\Pi_s = \frac{a_t-a_s}{2}X_s = (p-a_s)X_s \implies p=\frac{a_t-a_s}{2}.\]

\begin{proposition} \label{prop:same-p}
At equilibrium, the source offers the price $\frac{a_t-a_s}{2}$ to its buyers.
\end{proposition}

We prove the following propositions which show the location advantage of a dominating parent. We show that in the SS and the SM case, the seller benefits a lot from the competition among the buyer side, and the proof is provided in Appendix~\ref{app: cor: 4.1}.

\begin{restatable}{proposition}{proptwousssm}\label{cor: 4.1}
	In the SS and the SM case, the utility of the seller is at least twice the utility of the buyers' total utility.
\end{restatable}

If a firm controls all the flow of another child firm in the supply chain, then its utility is at least twice as much as its child. We refer the proof to Appendix~\ref{app: thm: 4.4}.

\begin{restatable}{proposition}{proptwou}\label{thm: 4.4}
	If firm $i$ is a dominating parent of firm $j$, then firm $i$ has at least twice the utility of firm $j$.
\end{restatable}

To sum up, a dominating parent always has better utility and the double utility rule always holds, which demonstrates the great value of controlling the upstream trades.


\subsection{Network Efficiency and Component Composition}

To measure how firms would benefit from the network, we may care about not only the flow value but also the {\em social welfare}. The social welfare is the total utility of the source and intermediary firms plus the consumer surplus:

\begin{equation} \label{eq: sw1}
\begin{aligned}
SW(G) &= \sum_{i \in V \cup \{s\}}{\Pi_i} + \frac{1}{2}b_t X_t^2\\
&= \frac{1}{2}\sum_{i \in V \cup \{s\}}{(b_i + \sum_{k \in C_P(i)} b_k) X_i^2} + \frac{1}{2}b_t X_t^2.
\end{aligned}
\end{equation}

The social welfare can also be interpreted as the product of the flow and the price difference between the sink and the source (the producer surplus), plus the consumer surplus:

\begin{equation} \label{eq: sw2}
\begin{aligned}
SW(G) &= (a_t - a_s - b_t X_s)X_s + \frac{1}{2}b_t X_t^2 \\
&= (a_t - a_s - \frac{b_t}{2} X_s)X_s.
\end{aligned}
\end{equation}

The criteria of interest are the welfare efficiency and flow efficiency defined as follows.

\begin{definition}[Welfare Efficiency]
	 A supply chain network is more welfare efficient if it provides a larger social welfare at equilibrium.
\end{definition}

\begin{definition}[Flow Efficiency]
	 A supply chain network is more flow efficient if it provides a larger flow at equilibrium.
\end{definition}

We examine the relation between flow efficiency, welfare efficiency, and the  structure of an SPG. Suppose the given SPG $G$ is constructed by series and parallel compositions on SPGs $G_1$, $G_2$, ..., and $G_n$, then $G_i$ where $i=1, ..., n$ are the \emph{components} of $G$. Since we assume that $G$ is shortcut-free, there are no shortcuts in the components of $G$ as well. The flow efficiency and welfare efficiency of a supply chain is highly related to its components, and we define the {\em component factor} as the following.
\begin{definition}[Component Factor]
	The \emph{component factor} of an SPG $Y$ is
\[\lambda(Y) := \frac{b_{s_Y}}{b_{t_Y}}.\]
\end{definition}

The component factor measures the enlargement of the decreasing linear coefficient at the equilibrium price of an SPG $Y$. $Y$ will have a higher flow value if its component factor $\lambda(Y)$ is small. The following theorem shows that the component factor is irrelevant to $b_{t_X}$. The proof is provided in Appendix~\ref{app: lem: 4.1}.

\begin{restatable}{lemma}{lemcp}\label{lem: 4.1}
$\lambda(Y) \geqslant 2$ is a constant that is only relevant to the graph structure of $Y$.
\end{restatable}

Now we can rewrite social welfare in terms of $\lambda(G)$
\begin{equation} \label{eq: sw3}
\begin{aligned}
SW(G) &= (a_t - a_s - \frac{b_t}{2} X_s)X_s \\
&= (a_t - a_s - \frac{b_t}{2} \frac{a_t - a_s}{\lambda(G)b_t})\frac{a_t - a_s}{\lambda(G)b_t} \\
&= (a_t - a_s)^2(1-\frac{1}{2\lambda(G)})\frac{1}{\lambda(G)b_t}.
\end{aligned}
\end{equation}
With fixed sink price function $a_t - b_t X_t$ and source production cost $a_s$, $SW(G)$ is a function of $\lambda(G)$. By Lemma \ref{lem: 4.1}, $\lambda(Y) \geqslant 2$, so $SW(G)$ is maximized when $\lambda(G)=2$, and it is a decreasing function of $\lambda(G)$ when $\lambda(G) \in [2,\infty)$. The flow $X_s$ is also a decreasing function of $\lambda(G)$. $G$ is the single edge network if and only if $\lambda(G)=2$. Therefore, the single edge network is the most flow and welfare efficient.
\begin{proposition}
The single edge network is the most flow and welfare efficient network.
\end{proposition}

If the network is fixed, we have the following. The proof is provided in Appendix~\ref{app: prop: 4.1}.

\begin{restatable}{proposition}{propinc}\label{prop: 4.1}
	With a fixed network structure, if the demand at the market increases or the cost at the source decreases, then the market is more flow and welfare efficient and the utility of each individual firm increases.
\end{restatable}

Consider the order of series composition on two SPGs $Y$ and $Z$. By Lemma~\ref{lem: 4.1},
\[
\lambda(S(Y,Z)) = \lambda(Y)\lambda(Z) = \lambda(S(Z,Y))\\
\]
which implies that swapping the order of two components in a series composition does not change the component factor.

\begin{lemma} \label{prop:SYZ}
Given two SPGs $Y$ and $Z$, $\lambda(S(Y,Z)) = \lambda(S(Z,Y))$.
\end{lemma}

We also present a closed-form expression of the component factor $\lambda(P(Y,Z))$ after the parallel composition in terms of $\lambda(Y)$ and $\lambda(Z)$, assuming that $P(Y,Z)$ is shortcut-free. We refer the proof of Lemma \ref{lem: 4.z} to Appendix~\ref{app: lem: 4.z}.

\begin{restatable}{lemma}{lempc}\label{lem: 4.z}
	Given SPGs $Y$ and $Z$ such that $P(Y,Z)$ does not have any shortcuts,
\[
\lambda(P(Y,Z)) = \frac{(\lambda(Y)-2)(\lambda(Z)-2)}{\lambda(Y)+\lambda(Z)-4} + 2.
\]
\end{restatable}

By Lemma \ref{prop:SYZ} and \ref{lem: 4.z}, during the construction of an SPG, the component factor remains unchanged after a parallel composition of two SPGs with unchanged component factor. Switching the order of two local components in a series composition does not change the global component factor. Hence with fixed sink price function and production cost, by \eqref{eq: sw3}, the flow and welfare efficiency remain unchanged after the swap. Since the flow remain the same and by Proposition~\ref{prop:same-p}, the source offers the same price to its buyers regardless of the network structure, the source utility is unchanged. We conclude by the following proposition.

\begin{proposition}
Suppose the given SPG $G$ is constructed by series and parallel compositions on SPGs $G_1$, $G_2$, ..., and $G_n$, which includes a step $S(G_i,G_j)$ for $i \neq j$ and $i,j \in \{1, ..., n\}$. Let $G'$ be an SPG with the same construction as $G$, except that $S(G_i,G_j)$ is replaced by $S(G_j,G_i)$. Then, $\lambda(G)=\lambda(G')$. With the same sink price function and production cost, the flow efficiency, welfare efficiency, and source utility remain the same.
\end{proposition}

\section{Equilibrium in Generalized Series Parallel Graphs} \label{sec: ext}

We discuss the equilibrium properties in the extension cases when the series parallel graph has multiple sources or sinks. In particular, we show that:
\begin{itemize}
    \item Single-source-and-multiple-sinks SPG: Price function of a firm may be piecewise linear and discontinuous under simple settings. There may exist multiple equilibria.
	\item Multiple-sources-and-single-sink SPG: An equilibrium may not exist.
\end{itemize}

\subsection{Single Source and Multiple Sinks} \label{sec: ssms}

A series parallel graph with a single source and multiple sinks (SMSPG) is defined as follows.

\begin{definition}[SMSPG] \label{def: 5.2}
$G$ is a {\em single-source-and-multiple-sink SPG} if it can be constructed by deleting the sink node of an SPG and setting the adjacent nodes of the sink as the new sink nodes. The set of sinks is denoted as $T$.
\end{definition}

First we consider a special case that all the markets have the same demand, then all the markets are {\em active} at equilibrium, i.e. every market has positive incoming flow. The proof is similar to Theorem~\ref{thm: 3.1} and we provide the sketch in Appendix~\ref{app: prop: 4.z}.

\begin{restatable}{proposition}{propsmspg}\label{prop: 4.z}
Given an SMSPG, if all the markets have the same demand, then there exists a unique equilibrium that can be found in polynomial time.
\end{restatable}

With different demand $a_t$ for $t \in T$, the markets may be {\em inactive}, i.e. the incoming quantity is zero. For example:

\begin{example} [Markets Activities]
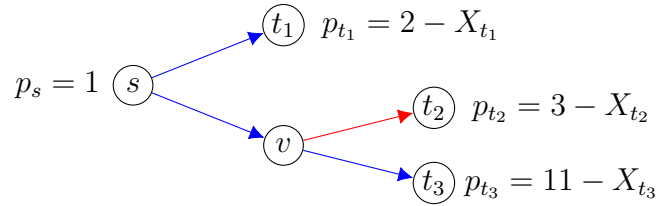
\mbox{}\\
\begin{center}
\begin{tikzpicture}
	\node[vertex] (s) at (-2,0) {$s$}; 
	\node[vertex] (t1) at (0,0.8) {$t_1$}; 
	\node[vertex] (a) at (0,-0.8) {$v$}; 
	\node[vertex] (t2) at (2,-.3) {$t_2$}; 
	\node[vertex] (t3) at (2,-1.3) {$t_3$};
	\node at (-3, 0) {$p_{s} = 1$};
	\node at (1.7, 0.8) {$p_{t_1} = 2 - X_{t_1}$};
	\node at (3.7, -.3) {$p_{t_2} = 3 - X_{t_2}$};
	\node at (3.7, -1.3) {$p_{t_3} = 11 - X_{t_3}$};
	\path[->]
		(s) edge [blue] (t1)
		(s) edge [blue] (a)
		(a) edge [red] (t2)
		(a) edge [blue] (t3);
\end{tikzpicture}
\end{center}
Since $a_{t_1} > p_s$, the market $t_1$ is active. Suppose markets $t_2$ and $t_3$ are both active at equilibrium, and $s$ offers $v$ $x+y$ units of goods with price $p_v$, then $v$ will buy all the goods from $s$. Suppose $v$ sells $x$ to $t_2$ and $y$ to $t_3$, then the utility of $v$ is
\[\Pi_v = (3-x)x+(11-y)y-p_v(x+y).\]
By taking the derivative of $\Pi_v$ with respect to $x$ and $y$, we have
\[\frac{\partial \Pi_v}{\partial x}=3-2x-p_v \text{ and } \frac{\partial \Pi_v}{\partial y}=11-2y-p_v\]
so $\Pi_v$ is maximized when
\[p_v = 3-2x = 11-2y \implies p_v = 7-(x+y)=7-X_v.\]

When source $s$ makes a decision, the flow $x_{st_1}$ and $x_{sv}$ can be handled independently, the optimal decision that maximizes the utility $(7-X_v)X_v$ of $s$ from $v$ is $X_v = 3.5$ and $p_v = 3.5 > a_{t_2}=3$, which contradicts to the assumption that market $t_2$ is active. Therefore, market $t_2$ is inactive, even though $a_{t_2} > a_{t_1}$.
\end{example}

The above example is against the intuition that the market with higher demand is more likely to be active ($t_2$ is inactive while $t_1$ is). While the truth is not only market demand, but also the competitors and network structure have influence on market activity. Namely, market $t_2$ is inactive because it has a longer supply chain than $t_1$ and a strong competition between $t_3$. As a result, it is less favorable than $t_1$ and $t_3$.



The market behavior of SMSPG is usually intractable. In particular, we focus on supply chain networks of the shape in Figure~\ref{fig: 5.0}. Based on the activity status of the markets, we introduce two types of strategies for the upstream firm.

\begin{definition}[Low Price Strategy]
	Firm processes relatively large quantity of goods at a relatively low price, such that all the markets are active.
\end{definition}
\begin{definition}[High Price Strategy]
	Firm processes relatively small quantity of goods at a relatively high price, such that some markets are inactive.
\end{definition}

The firm plays its strategy to maximize its utility. Because of various choice of strategies, the price function might be piecewise linear and discontinuous. Furthermore, some counterintuitive results will occur, i.e. the increase of demand may result in the decrease of total flow and social welfare (comparing to Proposition~\ref{prop: 4.1}). To understand these differences, it is helpful to consider an example as in Figure~\ref{fig: 5.0}, where the two supply chain networks have identical structure with different market demands. 

\begin{figure}[H]
\begin{center}
\begin{tikzpicture}
	\node at (-1, 1) {supply chain 1:};
	\node[vertex] (b) at (0,0) {$s$}; 
	\node at (-1, 0) {$p_s = 7$};
	\node[vertex] (a) at (2,0) {$v$}; 
	\node[vertex] (t1) at (4,.5) {$t_1$}; 
	\node[vertex] (t2) at (4,-.5) {$t_2$}; 
	\node [red] at (5.6, .5) {$p_{t_1} = 19 - x$};
	\node at (5.6, -.5) {$p_{t_2} = 12 - y$};
	\path[->]
		(b) edge [blue] node [above] {$\textcolor{black}{X_v}$} (a) 
		(a) edge [blue] node [above] {$\textcolor{black}{x}$} (t1) 
		(a) edge [blue] node [below] {$\textcolor{black}{y}$} (t2);
\end{tikzpicture}
\end{center}

\begin{center}
\begin{tikzpicture}
	\node at (-1, 1) {supply chain 2:};
	\node[vertex] (b) at (0,0) {$s$}; 
	\node at (-1, 0) {$p_s = 7$};
	\node[vertex] (a) at (2,0) {$v$}; 
	\node[vertex] (t1) at (4,.5) {$t_1$}; 
	\node[vertex] (t2) at (4,-.5) {$t_2$}; 
	\node [red] at (5.6, .5) {$p_{t_1} = 20 - x$};
	\node at (5.6, -.5) {$p_{t_2} = 12 - y$};
	\path[->]
		(b) edge [blue] node [above] {$\textcolor{black}{X_v}$} (a) 
		(a) edge [blue] node [above] {$\textcolor{black}{x}$} (t1) 
		(a) edge [red] node [below] {$\textcolor{black}{y}$} (t2);
\end{tikzpicture}
\end{center}
\caption{Multiple Sinks Supply Networks} \label{fig: 5.0}
\end{figure}
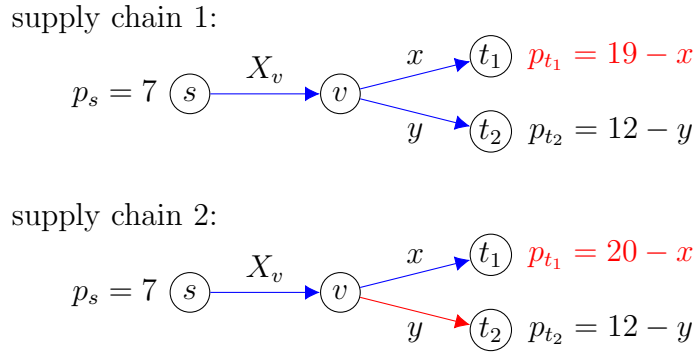

Intuitively, supply chain 2 with higher market demand should have larger flow and social welfare. However, supply chain 1 is more flow and welfare efficient. The equilibrium price functions at $s$ and $v$ are shown in Figure~\ref{fig: 5.1}. We note that the source firm $s$ has two strategies when $p_s = 7$, and both low and high price strategies are feasible. Interestingly, when $a_{t_1} = 20$, the utility of $s$ is maximized by choosing high price strategy and only market $t_1$ is active. However, when demand at market $t_1$ drops, the low price strategy is preferred by $s$.

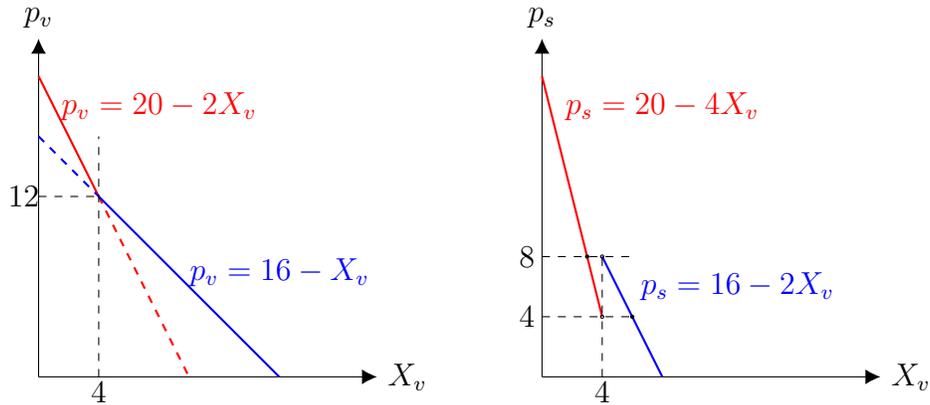
\begin{figure}[H]
\centering
\begin{subfigure}{.4\textwidth}
  \begin{center}
		\begin{tikzpicture}
			\draw[->] (0,0) -- (4.5,0) node[right] {$X_v$};
			\draw[->] (0,0) -- (0,4.5) node[above] {$p_v$};
			\draw[scale=0.2,domain=0:4,thick,variable=\x,red] plot ({\x},{20 - 2*\x}) 
			node at (8, 18) {$p_v = 20 - 2X_v$};
			\draw[scale=0.2,domain=4:10,thick,variable=\x,red,dashed] plot ({\x},{20 - 2*\x});

			\draw[scale=0.2,domain=4:16,thick,variable=\x,blue] plot ({\x},{16 - \x}) 
			node[above=30pt] {$p_v = 16 - X_v$};
			\draw[scale=0.2,domain=0:4,thick,variable=\x,blue,dashed] plot ({\x},{16 - \x});
			\draw [scale=0.2,dashed] (4,0) -- (4,16) node at (4, -1) {$4$};

			\draw [scale=0.2,dashed] (0,12) -- (4,12) node at (-1, 12) {$12$};
		\end{tikzpicture}
	\end{center}
\end{subfigure}
\begin{subfigure}{.4\textwidth}
  \begin{center}
		\begin{tikzpicture}
			\draw[->] (0,0) -- (4.5,0) node[right] {$X_v$};
			\draw[->] (0,0) -- (0,4.5) node[above] {$p_s$};
			\draw[scale=0.2,domain=0:4,thick,variable=\x,red] plot ({\x},{20 - 4*\x}) 
			node at (8, 18) {$p_s = 20 - 4X_v$};

			\draw[scale=0.2,domain=4:8,thick,variable=\x,blue]  plot ({\x},{16 - 2*\x}) 
			node at (13, 6) {$p_s = 16 - 2X_v$};

			\draw [scale=0.2,dashed] (4,0) -- (4,8) node at (4, -1) {$4$};

			\draw [scale=0.2] (4,4) circle (3pt);
			\draw [scale=0.2] (4,8) circle (3pt);
			\filldraw [scale=0.2] (3,8) circle (3pt);
			\filldraw [scale=0.2] (6,4) circle (3pt);

			\draw [scale=0.2,dashed] (0,4) -- (6,4) node at (-1, 4) {$4$};
			\draw [scale=0.2,dashed] (0,8) -- (6,8) node at (-1, 8) {$8$};
		\end{tikzpicture}
	\end{center}
\end{subfigure}
\caption{Piecewise Linear Price Functions of Supply Chain 2}
\label{fig: 5.1}
\end{figure}

By fixing the demand at market $t_2$ and adjusting the demand at market $t_1$ ($a_{t_1}$), Figure~\ref{fig: 5.2} shows the results of the source utility, consumer surplus, total flow and social welfare. The intersecting point at $a_{t_1} \approx 19.07$ shows that increasing demand at market $t_1$ hurts the supply chain efficiency. When $a_{t_1}$ is the intersecting point then there are multiple equilibria since $s$ has no preference between the high price and low price strategy. Besides, $a_{t_1}$ is feasible only in the interval $(12,22]$. The calculation details are provided in Example~\ref{app: ex: me_smspg} in the appendix.

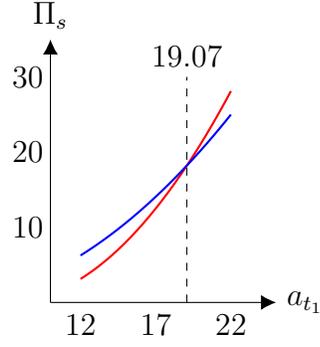
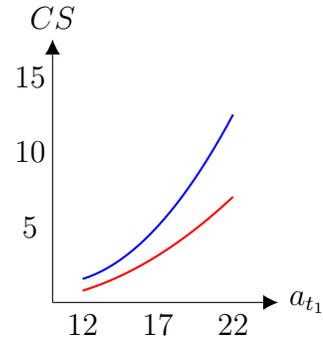
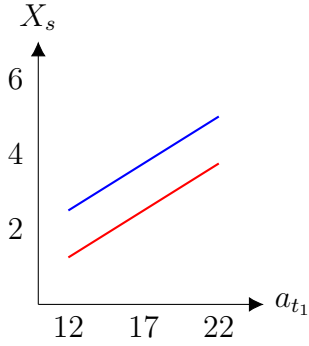
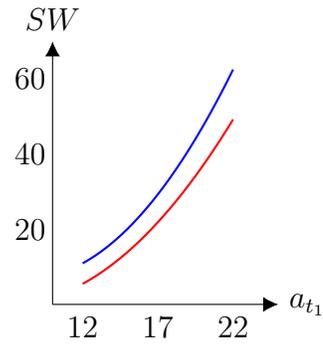
\begin{figure}[H]
			\centering
        \begin{subfigure}[b]{0.4\textwidth}
				\centering
            \begin{tikzpicture}
			\draw[->] (0,0) -- (3,0) node[right] {$a_{t_1}$};
			\draw[->] (0,0) -- (0,3.5) node[above] {$\Pi_s$};
			\draw[scale=0.2,domain=2:12,thick,variable=\x,red] plot ({\x},{((\x+10-7)/2*(\x+10-7)/4)/2});
			\draw[scale=0.2,domain=2:12,thick,variable=\x,blue] plot ({\x},{((\x+10-2)/4*(\x+10-2)/4)/2});

			\draw node at (0.4, -0.3) {$12$};
			\draw node at (1.4, -0.3) {$17$};
			\draw node at (2.4, -0.3) {$22$};

			\draw node at (-0.3,1) {$10$};
			\draw node at (-0.3,2) {$20$};
			\draw node at (-0.3,3) {$30$};
			\draw [dashed] (0.2*9.07,0) -- (0.2*9.07,3) node[above] {$19.07$};
		\end{tikzpicture}
			\caption[]%
					{Source Utility v.s. $a_{t_1}$}
        \end{subfigure}
        \quad
        \begin{subfigure}[b]{0.4\textwidth}  
            \centering
            \begin{tikzpicture}
			\draw[->] (0,0) -- (3,0) node[right] {$a_{t_1}$};
			\draw[->] (0,0) -- (0,3.5) node[above] {$CS$};
			\draw[scale=0.2,domain=2:12,thick,variable=\x,red] plot ({\x},{(((\x+10-7)/4)^2/2)});
			\draw[scale=0.2,domain=2:12,thick,variable=\x,blue] plot ({\x},{(((3/8*(\x+10)-13/4)^2+(11/4-(\x+10)/8)^2)/2)});

			\draw node at (0.4, -0.3) {$12$};
			\draw node at (1.4, -0.3) {$17$};
			\draw node at (2.4, -0.3) {$22$};

			\draw node at (-0.3,1) {$5$};
			\draw node at (-0.3,2) {$10$};
			\draw node at (-0.3,3) {$15$};
		\end{tikzpicture}
			\caption[]%
					{Consumer Surplus v.s. $a_{t_1}$}
        \end{subfigure}
        \vskip\baselineskip
        \begin{subfigure}[b]{0.4\textwidth}   
            \centering
            \begin{tikzpicture}
			\draw[->] (0,0) -- (3,0) node[right] {$a_{t_1}$};
			\draw[->] (0,0) -- (0,3.5) node[above] {$X_s$};
			\draw[scale=0.2,domain=2:12,thick,variable=\x,red] plot ({\x},{(\x+10-7)/1.6});
			\draw[scale=0.2,domain=2:12,thick,variable=\x,blue] plot ({\x},{(\x+10-2)/1.6});

			\draw node at (0.4, -0.3) {$12$};
			\draw node at (1.4, -0.3) {$17$};
			\draw node at (2.4, -0.3) {$22$};

			\draw node at (-0.3,1) {$2$};
			\draw node at (-0.3,2) {$4$};
			\draw node at (-0.3,3) {$6$};
		\end{tikzpicture}
			\caption[]%
					{Flow v.s. $a_{t_1}$}
        \end{subfigure}
        \quad
        \begin{subfigure}[b]{0.4\textwidth}   
            \centering
            \begin{tikzpicture}
			\draw[->] (0,0) -- (3,0) node[right] {$a_{t_1}$};
			\draw[->] (0,0) -- (0,3.5) node[above] {$SW$};
			\draw[scale=0.2,domain=2:12,thick,variable=\x,red] plot ({\x},{((\x+10-7)/2*(\x+10-7)/4)*7/4/4});
			\draw[scale=0.2,domain=2:12,thick,variable=\x,blue] plot ({\x},{(((3/8*(\x+10)-13/4)^2+(11/4-(\x+10)/8)^2)*3/2+((\x+10-2)/4)^2)/4});

			\draw node at (0.4, -0.3) {$12$};
			\draw node at (1.4, -0.3) {$17$};
			\draw node at (2.4, -0.3) {$22$};

			\draw node at (-0.3,1) {$20$};
			\draw node at (-0.3,2) {$40$};
			\draw node at (-0.3,3) {$60$};
		\end{tikzpicture}
			\caption[]%
					{Social Welfare v.s. $a_{t_1}$}
        \end{subfigure}
        \caption[]
        {High Price Strategy (red) v.s. Low Price Strategy (blue)} 
        \label{fig: 5.2}
    \end{figure}

\begin{proposition}
An SMSPG may have multiple equilibria.
\end{proposition}

\begin{remark} \label{rm1}
For supply chain networks of the shape in Figure~\ref{fig: 5.0},

\begin{itemize}
\item The low price strategy always gives a higher flow value than the high price strategy.
\item When the demand difference between two markets is small enough, the low price strategy gives better utility for the source. When the difference is large enough, the high price strategy gives better utility for the source.
\item The low price strategy always produces higher social welfare.
\end{itemize}
\end{remark}

In short, the low price strategy is preferred by $s$ if the demand difference is not large. Besides, with the low price strategy, everyone is usually better off. We refer more interpretation and detailed calculations of these results to Appendix~\ref{app: rm1}.

\subsection{Multiple Sources and Single Sink} \label{sec: msss}
The extension of SPG with multiple sinks is defined similarly to Definition~\ref{def: 5.2}.

\begin{definition}[MSSPG] \label{def: 5.1}
$G$ is a {\em multiple-source-and-single-sink SPG} if it can be constructed by deleting the source node of an SPG and setting the adjacent nodes of the source as the new source nodes. The set of sources is denoted as $S$.
\end{definition}

Assume that the source producers make their decision simultaneously, an equilibrium may not exist. We show Example~\ref{ex: SMSPG} in the appendix.

\begin{proposition}
    An equilibrium in an MSSPG may not exist.
\end{proposition}

\section{Conclusion} \label{sec:conclude}
We consider a model of sequential competition in supply chain networks. Our main contribution is that when the network is series parallel, the model is tractable and allows a rich set of comparative analysis. In particular, we provide a linear time algorithm to compute the equilibrium and the algorithm helps us study the influence of the network to the total flow and social welfare of the equilibrium.

Slightly extending the network structure beyond series parallel graphs with a single source and multiple sinks (SMSPG) makes the model intractable. The first open problem is to design an efficient algorithm that verifies if an SMSPG has an equilibrium and finds one if it exists. The main challange is the piecewise-linearity and the discontinuity of the price function for intermediary firms. This problem may be computationally intractable but it is unclear what a reasonable proof strategy would be.

Another open problem is to efficiently find an equilibrium in general DAGs with a single source and a single sink. The trades can be inactive for shortcut-free DAGs as shown in Example~\ref{app: ex: gen_dag} in the appendix. We conjecture that there is always an equilibrium and the active trades form a shortcut-free SPG. A natural approach is to compute the price function in a reverse topological order from the sink inductively. However, this enforces one to solve LCPs that correspond to the firms, where the LCPs of the upstream firms are derived from the LCPs of the downstream firms. Solving the LCP system requires determining inactive trades where the number of combination of active and inactive trades is exponential. Therefore, a potential strategy to show computational intractability is a reduction from an LCP based or a quadratic programming problem.
 
\setstretch{1}
\bibliographystyle{abbrvnat}
\bibliography{ref}

\setstretch{1}
\newpage
\section*{Appendix}
\appendix

\section{Proofs in Section \ref{sec: clearance}}

\subsection{Proof of The Nonexistence of the MM Case} \label{app:no_mm}
\begin{proof}
	The Multiple sellers and multiple buyers (MM) case is $|B(i)| \geqslant 2$ and $|S(j)| \geqslant 2$ for $ij \in E$:
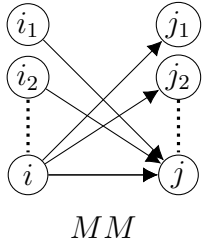
	\begin{center}
	\begin{tikzpicture}
		\node[vertex] (i1) at (2,1.5) {$i_1$}; 
		\node[vertex] (i2) at (2,0.8) {$i_2$};
		\node[vertex] (i) at (2,-.5) {$i$};
		\node[vertex] (j1) at (4,1.5) {$j_1$}; 
		\node[vertex] (j2) at (4,0.8) {$j_2$};
		\node[vertex] (j) at (4,-.5) {$j$}; 
		\draw[dotted, very thick] (2,0.5) -- (2,-.2);
		\draw[dotted, very thick] (4,0.5) -- (4,-0.2);
		\path[->]
			(i) edge (j)
			(i1) edge (j)
			(i2) edge (j)
			(i) edge (j1)
			(i) edge (j2)
			(i) edge (j);
		\node at (3, -1.2) {$MM$};
	\end{tikzpicture}
	\end{center}
    
    The MM case is impossible in a shortcut-free SPG, and this can be proved by induction. By definition, any SPG can be constructed by series and parallel composition:
    \begin{itemize}
        \item Series composition: the MM case will not appear after series composition.
        \item Parallel composition: by checking the merging source and sink, one can see that the MM case will not appear either, unless there is a shortcut between the source and sink.
    \end{itemize}

    Therefore, the MM case does not happen in a shortcut-free SPG. \QED
\end{proof}

\subsection{Proof of Theorem~\ref{lem:pq-relation}} \label{app:lem:pq-relation}
\lemeqprice*
\begin{proof}
	Our strategy is to start from the sink $t$ and argue inductively via reverse topological traversal that the price proposed to $i$ must be an affine decreasing function of $X_i$. We note that the computation is always under the flow reservation condition by Observation~\ref{obs: fc}.

\paragraph{Starting from the Sink $t$.}

We start with the behavior of the direct upstream firms of sink $t$. For a firm $i \in S(t)$, if arc $it$ belongs to the SS case, then the utility of $i$ is
\begin{align*}
\Pi_i &= (a_t - b_t x_{it}) x_{it} - \sum_{k \in S(i)} p_{ki} x_{ki} \\
&= (a_t - b_t \sum_{k \in S(i)} x_{ki}) \sum_{k \in S(i)} x_{ki} - \sum_{k \in S(i)} p_{ki} x_{ki}.
\end{align*}

$p_{ki}$ given by the selling firms are regarded as constants to $i$. $\Pi_i$ is a concave function. By taking the derivative of $\Pi_i$ with respect to $x_{ki}$, we have
\[\frac{\partial \Pi_i}{\partial x_{ki}} = a_t - 2b_t \sum_{k \in S(i)} x_{ki} - p_{ki}.\]

The price and quantity at equilibrium is a solution of the following linear complementarity problem. Intuitively, $\frac{\partial \Pi_i}{\partial x_{ki}} > 0$ cannot happen since otherwise $k$ could have raised the price $p_{ki}$ such that $\frac{\partial \Pi_i}{\partial x_{ki}} = 0$. This makes $i$ accept all the goods from $k$ and $k$ would obtain a higher payoff, which contradicts the equilibrium condition. If $\frac{\partial \Pi_i}{\partial x_{ki}} < 0$, then $k$ will not offer any goods to $i$ so $x_{ki}=0$ since the payoff of $i$ will decrease if $i$ accepts some goods from $k$.
\begin{equation} \label{eq: r-lcp} \tag{Reverse LCP}
\begin{cases}
\sum_{k \in S(i)}\frac{\partial \Pi_i}{\partial x_{ki}} x_{ki}  = 0,\\
\frac{\partial \Pi_i}{\partial x_{ki}} \leqslant 0  \quad \forall k \in S(i),\\
x_{ki} \geqslant 0 \quad \forall k \in S(i).
\end{cases}
\end{equation}

When $ki$ is active, $p_{ki}$ is such that $\frac{\partial \Pi_i}{\partial x_{ki}} = 0$. Therefore, for an active arc $ki$,
\[p_{ki} = p_i = a_t - 2b_t \sum_{k \in S(i)} x_{ki}.\]

If arc $it$ belongs to the MS case, then the utility of $i$ is
\[\Pi_i = (a_t - b_t \sum_{j \in S(t)} x_{jt}) \sum_{k \in S(i)} x_{ki} - \sum_{k \in S(i)} p_{ki} x_{ki}.\]

$p_{ki}$ and $x_{jt}$ where $j \neq i$ are regarded as constants to $i$ so that regardless of some fixed $p_{ki}$ and $x_{jt}$, $i$ is not willing to change its decision at equilibrium. $\Pi_i$ is a concave function. By taking the derivative of $\Pi_i$ with respect to $x_{ki}$, we have
\[\frac{\partial \Pi_i}{\partial x_{ki}} = a_t - b_t \sum_{k \in S(i)} x_{ki} - b_t \sum_{j \in S(t)} x_{jt} - p_{ki}.\]

By a similar argument as before, the price and quantity at equilibrium is a solution of \ref{eq: r-lcp}. When $ki$ is active, $p_{ki}$ is such that $\frac{\partial \Pi_i}{\partial x_{ki}} = 0$. Therefore, for an active arc $ki$,
\[p_{ki} = p_i = a_t - b_t \sum_{k \in S(i)} x_{ki} - b_t \sum_{j \in S(t)} x_{jt} = a_t - b_i X_i - b_t X_t.\]

\paragraph{Before reaching the SM Case.}
The same procedure as before can be inductively repeated whenever we meet an MS or SS case by the reverse topological traversal from $t$. Given the fact that the downstream price must be an affine decreasing function of the inflow of the parent merging child nodes, the derivative of the firm utility with respect to the quantity decision variables must always be zero whenever the quantity is positive.

Before reaching the SM case during the reverse topological traversal from $t$, consider firm $i \in V$. By inductive hypothesis, suppose for each arc $x_{ij} > 0$,
\[p_{ij} = p_j = a_t - b_j X_j - \sum_{l \in C_P(j)} b_l X_l.\]

We note that the flow on $ki$ merges to the nodes $l \in C_P(j)$ and $j$, so $\frac{\partial X_j}{\partial x_{ki}} = \frac{\partial X_l}{\partial x_{ki}} = 1$.

Consider arc $ki$, if $ki$ is the SS case or the MS case, then the utility of $i$ is
\begin{align*}
\Pi_i &= p_j x_{ij} - \sum_{k \in S(i)} p_{ki} x_{ki} \\
&= p_j \sum_{k \in S(i)} x_{ki} - \sum_{k \in S(i)} p_{ki} x_{ki}.
\end{align*}

If $ki$ is the SS case and $x_{ki} > 0$, then by \ref{eq: r-lcp}, $\frac{\partial \Pi_i}{\partial x_{ki}} = 0$, thus
	\begin{align}
	p_{ki} &= p_j + \frac{\partial p_j}{\partial x_{ki}} \sum_{k \in S(i)} x_{ki} \nonumber \\
	&= a_t - b_j X_j - \sum_{l \in C_P(j)} b_l X_l - (b_j + \sum_{l \in C_P(j)} b_l) \sum_{k \in S(i)} x_{ki} \nonumber \\
	&= a_t - (2 b_j + \sum_{l \in C_P(j)} b_l) X_i - \sum_{l \in C_P(i)} b_l X_l \label{ss-price} \tag{SS-price}
	\end{align}
	where $X_i = X_j = \sum_{k \in S(i)} x_{ki}$ and $C_P(i) = C_P(j)$ by Observation~\ref{obs:mc}.

If $ki$ is the MS case and $x_{ki} > 0$, then by \ref{eq: r-lcp}, $\frac{\partial \Pi_i}{\partial x_{ki}} = 0$, thus
\begin{align}
	p_{ki} &= p_j + \frac{\partial p_j}{\partial x_{ki}} \sum_{k \in S(i)} x_{ki} \nonumber \\
	&= a_t - b_j X_j - \sum_{l \in C_P(j)} b_l X_l - (b_j + \sum_{l \in C_P(j)} b_l) \sum_{k \in S(i)} x_{ki} \nonumber \\
	&= a_t - (b_j +\sum_{l \in C_P(j)} b_l) X_i - b_j X_j - \sum_{l \in C_P(j)} b_l X_l \nonumber \\
	&= a_t - (b_j + \sum_{l \in C_P(j)} b_l) X_i - \sum_{l \in C_P(i)} b_l X_l \label{ms-price} \tag{MS-price}
	\end{align}
	where $X_i = \sum_{k \in S(i)} x_{ki}$ and $C_P(i) = C_P(j) \sqcup \{j\}$ by Observation~\ref{obs:mc} in this case.

\paragraph{Reaching the SM Case.} \label{para:sm}
Define the set of nodes $N_G$ such that for any node $i \in N_G$, $i$ itself and all the children of $i$ all have an empty self-merging child nodes set. Formally,
\[
N_G = \{i \mid i \in V \cup\{s\} \text{ such that } \forall j \in C(i) \cup \{i\}, C_s(j) = \emptyset\}.
\]
$N_G$ denotes the set of nodes starting from $t$ via reverse topological traversal until we reach a set of nodes that are sellers in the SM case. These sellers can be defined as the set of nodes $S_G$, such that for any node $i \in S_G$, there exists a buyer of $i$ that belongs to $N_G$. Formally,
\[S_G = \{i \mid i \in V \cup \{s\} \text{ such that } B(i) \cap N_G \neq \emptyset\}.\]

In the following figure, $N_G$ consists of the red nodes $j_1$, $v_1$, $v_2$, $l$, $k$, and $t$, while $S_G$ consists of the black nodes $s$ and $j_2$.

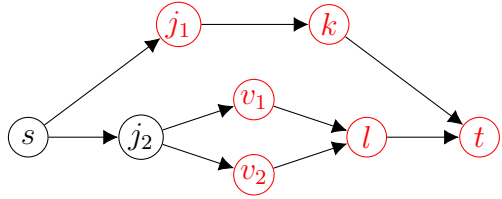
\begin{center}
\begin{tikzpicture}[baseline=0]
	\node[vertex] (s) at (0,0) {$s$}; 
	\node[vertex,red] (a) at (2,1.5) {$j_1$}; 
	\node[vertex] (a1) at (1.5,0) {$j_2$};
	\node[vertex,red] (b) at (3,0.5) {$v_1$}; 
	\node[vertex,red] (c) at (3,-0.5) {$v_2$};
	\node[vertex,red] (d) at (4.5,0) {$l$};
	\node[vertex,red] (f) at (4,1.5) {$k$};
  \node[vertex,red] (t) at (6,0) {$t$};
	\path[->]
		(s) edge (a)
		(s) edge (a1) 
		(a1) edge (b)
		(a1) edge (c)
		(a) edge (f)
		(b) edge (d)
		(c) edge (d)
		(f) edge (t)
		(d) edge (t)
;
\end{tikzpicture}
\end{center}

There must exist a node $i \in S_G$ such that all its buyers $j \in B(i)$ belong to $N_G$, and $i$ is the seller in the SM case. Given that $p_j$ where $j \in B(i)$ are affine decreasing functions, the utility of $i$ is
\begin{equation} \label{eq: pi_i}
\Pi_i = \sum_{j \in B(i)}{p_j x_{ij}} - \sum_{k \in S(i)} p_{ki} x_{ki}.
\end{equation}

Suppose the sellers of $i$ make their best decision and offer $i$ total inflow $X_i = C$ such that $i$ accepts everything. $i$ does not have control over the buying cost $\sum_{k \in S(i)} p_{ki} x_{ki}$ and the total inflow $\sum_{k \in S(i)} x_{ki}$ at equilibrium. What $i$ can decide is how to distribute $C$ to its buyers in $N_G$. Here $p_{ki}$ are regarded as constants given to $i$ and $\Pi_i$ is a concave function. There is also a constraint $\sum_{j \in B(i)}x_{ij} = \sum_{k \in S(i)}x_{ki}$. Therefore, we can rewrite the problem of maximizing $\Pi_i$ as the following convex quadratic program:

\begin{equation} \label{eq: dist} \tag{SM-CQP}
\begin{aligned}
& \maximize_{x} & & \sum_{j \in B(i)}{p_j x_{ij}} \\
& \text{subject to}
& & \sum_{j \in B(i)}{x_{ij}} = C.
\end{aligned}
\end{equation}

Consider the Lagrangian function:
	\[
	\begin{aligned}
	L(x, \lambda) &= \sum_{j \in B(i)}{p_j x_{ij}} - \lambda(\sum_{j \in B(i)}{x_{ij}} - C).
	\end{aligned}
	\]

By taking the derivative of $L(x,\lambda)$ with respect to $x_{ij}$, we have
	\begin{align}
	\frac{\partial L(x, \lambda)}{\partial x_{ij}} &= p_j + \sum_{h \in B(i)} \frac{\partial p_h}{\partial x_{ij}} x_{ih} - \lambda \nonumber \\
	&= a_t - b_j x_{ij} - \sum_{l \in C_P(j)} b_l X_l - b_j x_{ij} - \sum_{l \in C_T(i,j)} b_l X_l - \sum_{l \in C_P(i)} b_l X_i - \lambda \nonumber \\
	&= a_t - 2 b_j x_{ij} - 
2 \sum_{l \in C_T(i,j)} b_l X_l - \sum_{l \in C_P(i)} b_l (X_l + C) - \lambda \label{eq: dLdx}
	\end{align}
where $\sum_{l \in C_P(i)} b_l (X_l + C)$ is fixed by the variables $X_l$ and $C$ which are decided by the upstream buyers of $i$, thus regarded as a constant to $i$. The second equality follows by rearranging and summing the inflow value of the merging nodes and the inductive hypothesis on $p_j$.

$\sum_{j \in B(i)}{p_j x_{ij}}$ is maximized when $\frac{\partial L(x_{ij}, \lambda)}{\partial x_{ij}} = 0$ for each $j \in B(i)$. This indicates
\[
a_t - 2 b_j x_{ij} - 
2 \sum_{l \in C_T(i,j)} b_l X_l = \sum_{l \in C_P(i)} b_l (X_l + C) + \lambda.
\]
By rearranging, this can be formulated as a linear system
\begin{equation} \label{ls-sm} \tag{SM-LS}
\begin{cases}
2 b_j x_{ij} +
2 \sum_{l \in C_T(i,j)} b_l X_l = D \quad \forall j \in B(i), \\
\sum_{j \in B(i)}{x_{ij}} = C,
\end{cases}
\end{equation}
where $D:=a_t - \sum_{l \in C_P(i)} b_l (X_l + C) - \lambda$. Since the right hand side is the same for each linear constraint $j \in B(i)$, $x_{ij}=\alpha_{j} C$ where $\sum_{j \in B(i)} \alpha_{j}=1$ is the solution of \ref{eq: dist}. We focus on finding the \emph{convex coefficients} $\alpha_j$ and the price $p_i$.

\paragraph{Finding the Convex Coefficients and the Price.} \label{para:fccp} Consider the SM case where $ij \in E$. We start with the simple SM case  as a warm up and continue on the general SM case. \\

\textbf{Simple SM:} $|B(i)| \geqslant 2$, $|S(j)| = 1$, and $|C_S(i)|=1$:
	
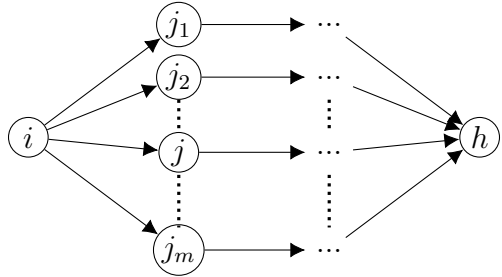
\begin{center}
	\begin{tikzpicture}
		\node[vertex] (j1) at (8,1.5) {$j_1$}; 
	\node[vertex] (j2) at (8,0.8) {$j_2$};
	\node[vertex] (jk) at (8,-1.5) {$j_m$}; 
	\node[vertex] (i) at (6,0) {$i$};
	\node[vertex] (j) at (8,-0.2) {$j$};
	\node (d1) at (10,1.5) {$...$};
	\node (d2) at (10,0.8) {$...$};
	\node (d3) at (10,-0.2) {$...$};
	\node (d4) at (10,-1.5) {$...$};
	\node[vertex] (h) at (12,0) {$h$};
	\draw[dotted, very thick] (8,0.5) -- (8,0.1);
	\draw[dotted, very thick] (8,-0.5) -- (8,-1.2);
	\draw[dotted, very thick] (10,0.5) -- (10,0.1);
	\draw[dotted, very thick] (10,-0.5) -- (10,-1.2);
	\path[->]
		(i) edge (j1)
		(i) edge (j2)
		(i) edge (j)
 		(i) edge (jk)
		(j1) edge (d1)
		(j2) edge (d2)
		(j) edge (d3)
 		(jk) edge (d4)
		(d1) edge (h)
		(d2) edge (h)
		(d3) edge (h)
 		(d4) edge (h);
		\end{tikzpicture}
	\end{center}

In the simple SM case, $C_T(i,j) = C_S(i) \setminus C_P(i)$ is the same for each $j \in B(i)$ so $b_j x_{ij}$, so $2 \sum_{l \in C_T(i,j)} b_l X_l$ is also the same. It suffices to find $\alpha_i$ so that $b_j x_{ij}$ is the same for each $j \in B(i)$. Let
\begin{equation} \label{app: eq: alpha_j} \tag{Simple-SM-$\alpha_j$}
\alpha_j = \frac{\frac{1}{b_j}}{\sum_{j' \in B(i)}{\frac{1}{b_{j'}}}},
\end{equation}
then when $x_{ij}=\alpha_j C$, $2 b_j x_{ij} +
2 \sum_{l \in C_T(i,j)} b_l X_l$ are the same for $j \in B(i)$ in \ref{ls-sm}. The best strategy for $i$ is to assign $\alpha_j C$ on arc $ij$. For all $j' \in B(i)$, $b_j'$ is positive so $\alpha_j$ is also positive. Therefore, if $C$ is positive, then $x_{ij}$ are all active.

Price $p_j$ is the same for each $j \in B(i)$:
\begin{align}
p_j &= a_t - b_j x_{ij} - \sum_{l \in C_p(j)}{b_l X_l} \nonumber \\
	&= a_t - b_j\frac{\frac{1}{b_j}}{\sum_{j' \in B(i)}{\frac{1}{b_{j'}}}}C - \sum_{l \in C_T(i,j)}{b_l X_l} - \sum_{l \in C_P(i)}{b_l X_l} \nonumber \\
	&= a_t - \frac{1}{\sum_{j' \in B(i)}{\frac{1}{b_{j'}}}}X_i - \sum_{l \in C_S(i) \setminus C_P(i)} b_l X_i - \sum_{l \in C_P(i)}{b_l X_l} \nonumber
\end{align}
where the last equality holds since $X_l = X_i$ when $l \in C_S(i) \setminus C_P(i)$. The utility of $i$ is
\begin{equation*}
\Pi_i = (a_t - \frac{1}{\sum_{j' \in B(i)}{\frac{1}{b_{j'}}}}X_i - \sum_{l \in C_S(i) \setminus C_P(i)} b_l X_i - \sum_{l \in C_P(i)}{b_l X_l})X_i - \sum_{k \in S(i)} p_{ki} x_{ki}.
\end{equation*}

By taking the derivative of $\Pi_i$ with respect to $x_{ki}$, we have
\begin{align*}
\frac{\partial \Pi_i}{\partial x_{ki}} &= a_t - \frac{1}{\sum_{j' \in B(i)}{\frac{1}{b_{j'}}}} X_i - \sum_{l \in C_S(i) \setminus C_P(i)} b_l X_i - \sum_{l \in C_P(i)}{b_l X_l}\\
& \quad - (\frac{1}{\sum_{j' \in B(i)}{\frac{1}{b_{j'}}}}+ \sum_{l \in C_S(i) \setminus C_P(i)} b_l + \sum_{l \in C_P(i)}{b_l})X_i - p_{ki} \\
&= a_t - (\frac{2}{\sum_{j' \in B(i)}{\frac{1}{b_{j'}}}}+ 2\sum_{l \in C_S(i) \setminus C_P(i)}{b_l}+\sum_{l \in C_P(i)}{b_l})X_i - \sum_{l \in C_P(i)}{b_l X_l} - p_{ki}.
\end{align*}

We note that $\frac{\partial X_i}{\partial x_{ki}} = 1$ and $\frac{\partial X_l}{\partial x_{ki}} = 1$ for $l \in C_P(i) \cup C_S(i)$. By \ref{eq: r-lcp}, if $x_{ki} > 0$, then $\frac{\partial \Pi_i}{\partial x_{ki}}=0$, so
\begin{equation} \label{simple-sm-price} \tag{Simple-SM-price}
p_{ki} = a_t - (\frac{2}{\sum_{j' \in B(i)}{\frac{1}{b_{j'}}}}+ 2\sum_{l \in C_S(i) \setminus C_P(i)}{b_l}+\sum_{l \in C_P(i)}{b_l})X_i - \sum_{l \in C_P(i)}{b_l X_l}.
\end{equation}

\textbf{General SM:} $|B(i)| \geqslant 3$, $|S(j)| = 1$, and $|C_S(i)| \geqslant 2$:

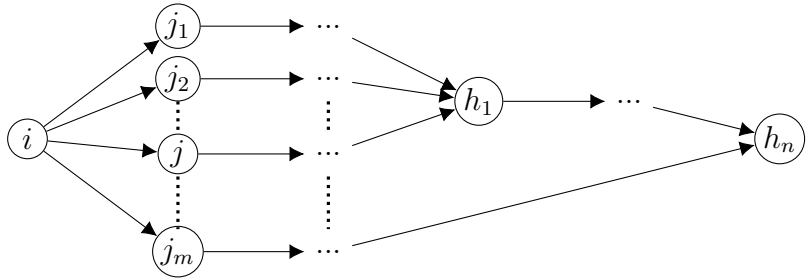
\begin{center}
	\begin{tikzpicture}
		\node[vertex] (j1) at (8,1.5) {$j_1$}; 
	\node[vertex] (j2) at (8,0.8) {$j_2$};
	\node[vertex] (jk) at (8,-1.5) {$j_m$}; 
	\node[vertex] (i) at (6,0) {$i$};
	\node[vertex] (j) at (8,-0.2) {$j$};
	\node (d1) at (10,1.5) {$...$};
	\node (d2) at (10,0.8) {$...$};
	\node (d3) at (10,-0.2) {$...$};
	\node (d4) at (10,-1.5) {$...$};
	\node (dd) at (14,0.5) {$...$};
	\node[vertex] (h) at (12,0.5) {$h_1$};
	\node[vertex] (h2) at (16,0) {$h_n$};
	\draw[dotted, very thick] (8,0.5) -- (8,0.1);
	\draw[dotted, very thick] (8,-0.5) -- (8,-1.2);
	\draw[dotted, very thick] (10,0.5) -- (10,0.1);
	\draw[dotted, very thick] (10,-0.5) -- (10,-1.2);
	\path[->]
		(i) edge (j1)
		(i) edge (j2)
		(i) edge (j)
 		(i) edge (jk)
		(j1) edge (d1)
		(j2) edge (d2)
		(j) edge (d3)
 		(jk) edge (d4)
		(d1) edge (h)
		(d2) edge (h)
		(d3) edge (h)
		(h) edge (dd)
		(dd) edge (h2)
 		(d4) edge (h2);
		\end{tikzpicture}
	\end{center}

List the nodes in $C_S(i)$ as $h_1, h_2, ..., h_n$ according to a topological order, where $h_1$ is the first merging node, $h_2$ is the second merging node, ..., $h_n$ is the last merging node in $C_S(i)$. For each $h_k \in C_S(i)$, let $B_k(i)$ be the largest subset of $B(i)$, such that for each $j \in B_k(i)$, $ij$ is the first arc of a corresponding disjoint path that eventually reaches $h_k$ without reaching any $h_r$ where $r < k$. Let $P_k(i)$ be the {\em direct merging parent set} consists of nodes $h_r$ that can reach $h_k$ without passing any $h_l$ where $r < l < k$.

In the following example, $C_S(i) = \{h_1, h_2\}$. The paths from $i$ via $j_2$ and $j_3$ merge at $h_1$. The path from $i$ via $j_1$ reaches $h_2$ without passing $h_1$. $B_2(i) = \{j_1\}$ and $B_1(i) = \{j_2, j_3\}$. $P_2(i) = \{h_1\}$ and $P_1(i) = \emptyset$.

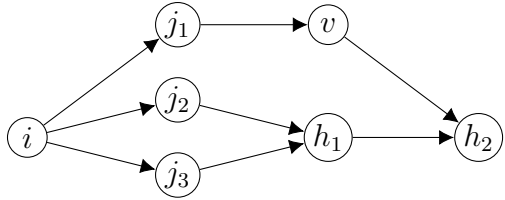
\begin{center}
\begin{tikzpicture}[baseline=0]
	\node[vertex] (s) at (0,0) {$i$}; 
	\node[vertex] (a) at (2,1.5) {$j_1$}; 
	\node[vertex] (b) at (2,0.5) {$j_2$}; 
	\node[vertex] (c) at (2,-0.5) {$j_3$};
	\node[vertex] (d) at (4,0) {$h_1$};
	\node[vertex] (f) at (4,1.5) {$v$};
  \node[vertex] (t) at (6,0) {$h_2$};
	\path[->]
		(s) edge (a) 
		(s) edge (b)
		(s) edge (c)
		(a) edge (f)
		(b) edge (d)
		(c) edge (d)
		(f) edge (t)
		(d) edge (t)
;
\end{tikzpicture}
\end{center}

	The calculation of $\alpha_i$ is done in an inductive fashion. We start from $k=1$, then $k=2$, and so on until $k=n$. We define an {\em aggregate variable} $c_k(i)$ for the vertices $B_k(i) \sqcup \{h_k\}$ recursively as the following:

\begin{equation} \label{app: eq: agg}
c_k(i):=	\frac{1}{\sum_{j \in B_k(i)} \frac{1}{b_{j}} + \sum_{l \mid h_l \in P_k(i)}\frac{1}{c_l(i)}} + b_{h_k}.
\end{equation}

When $P_k(i) = \emptyset$, we only consider the nodes in $B_k(i)$. When $h_n$ is reached, if $h_n \in C_P(i)$, then $b_{h_n}$ is not part of the aggregate variable $b_{B_n(i)}$. Therefore,
\begin{equation} \label{app: eq: agg-n}
	c_n(i) := \frac{1}{\sum_{j \in B_n(i)} \frac{1}{b_{j}} + \sum_{l \mid h_l \in P_n(i)}\frac{1}{c_l(i)}} + \sum_{l \in \{h_n\} \setminus C_P(i)}{b_l}.
\end{equation}

Now we find the convex coefficient for $ij$ where $j \in B(i)$, which allows us to rewrite $p_i$ in terms of $X_i$. The approach is a traversal of merging nodes until $h_n$ is reached. At step $k$, for each $p$ such that $h_p \in P_k(i)$, the aggregate variable $c_p(i)$ with nodes $j \in B_k(i)$ that merges to $h_k$ is {\em currently} weighted by
\begin{equation} \label{app: eq: aggcoeff}
	\beta_p(i) := \frac{\frac{1}{c_p(i)}}{\sum_{j \in B_k(i)} \frac{1}{b_{j}} + \sum_{l \mid h_l \in P_k(i)}\frac{1}{c_l(i)}}
\end{equation}
while node $j \in B_k(i)$ is weighted by
\begin{equation} \label{app: eq: vercoeff}
	\beta_{j} := \frac{\frac{1}{b_j}}{\sum_{j' \in B_k(i)} \frac{1}{b_{j'}} + \sum_{l \mid h_l \in P_k(i)}\frac{1}{c_l(i)}}.
\end{equation}

For $j \in B(i)$, the convex coefficient of $ij$ is
\begin{equation} \label{app: eq: prodcoeff} \tag{General-SM-$\alpha_j$}
	\alpha_j = \beta_{j}\prod_{p \mid h_p \in C_T(i,j) \setminus \{h_n\}}{\beta_p(i)}.
\end{equation}
We note that $\alpha_j$ is positive since $b_{j'}$ and $b_l$ where $j' \in B(i)$ and $l \in C_T(i,j')$ are all positive. When $x_{ij}=\alpha_j C$, $2 b_j x_{ij} +
2 \sum_{l \in C_T(i,j)} b_l X_l$ are the same for $j \in B(i)$ in \ref{ls-sm}. Namely, $x_{ij}=\alpha_j C$ is the solution of \ref{ls-sm}. This implies if $C > 0$, then $ij$ are all active. In particular, when $x_{ij}=\alpha_j C$,
\begin{equation} \label{cnx}
b_j x_{ij} + \sum_{l \in C_T(i,j)}{b_l X_l} = c_n(i) C = c_n(i) X_i.
\end{equation}

Price $p_j$ is the same for each $j \in B(i)$:
\begin{align}
p_j &= a_t - b_j x_{ij} - \sum_{l \in C_p(j)}{b_l X_l} \nonumber \\
	&= a_t - b_j x_{ij} - \sum_{l \in C_T(i,j)}{b_l X_l} - \sum_{l \in C_P(i)}{b_l X_l} \nonumber \\
	&= a_t - c_n(i) X_i - \sum_{l \in C_P(i)}{b_l X_l}. \label{sm-pj}
\end{align}

The utility of $i$ is
\[
\Pi_i = (a_t - c_n(i)X_i - \sum_{l \in C_P(i)}{b_l X_l})X_i - \sum_{k \in S(i)} p_{ki} x_{ki}.
\]

By taking the derivative of $\Pi_i$ with respect to $x_{ki}$, we have
\begin{align}
\frac{\partial \Pi_i}{\partial x_{ki}} &= a_t - c_n(i)X_i - \sum_{l \in C_P(i)}{b_l X_l} - (c_n(i) + \sum_{l \in C_P(i)}{b_l})X_i - p_{ki} \nonumber \\
&= a_t - (2c_n(i)+\sum_{l \in C_P(i)}{b_l})X_i - \sum_{l \in C_P(i)}{b_l X_l} - p_{ki}. \label{sm-dev-0}
\end{align}

We note that $\frac{\partial X_i}{\partial x_{ki}} = 1$ and $\frac{\partial X_l}{\partial x_{ki}} = 1$ for $l \in C_P(i)$. By \ref{eq: r-lcp}, if $x_{ki} > 0$, then $\frac{\partial \Pi_i}{\partial x_{ki}}=0$, so
\begin{equation} \label{general-sm-price} \tag{General-SM-price}
p_{ki}=a_t - (2c_n(i)+\sum_{l \in C_P(i)}{b_l})X_i - \sum_{l \in C_P(i)}{b_l X_l}.
\end{equation}

When $n=1$, the general SM case is exactly the simple SM case where $P_1(i)=\emptyset$ and $B_1(i)=B(i)$. By \eqref{app: eq: agg-n} and \ref{app: eq: prodcoeff},
\[c_1(i)=\frac{1}{\sum_{j \in B(i)}{\frac{1}{b_j}}}+ \sum_{l \in C_S(i) \setminus C_P(i)}{b_l} \text{ and } \alpha_j=\beta_j=\frac{\frac{1}{b_j}}{\sum_{j' \in B(i)}{\frac{1}{b_j'}}}.\]
\ref{app: eq: alpha_j} exactly matches \ref{app: eq: prodcoeff} and \ref{simple-sm-price} exactly matches \ref{general-sm-price}.

\paragraph{Reverse Topological Traversal until Reaching the Source.} We have shown that when the price function of $t$ is an affine decreasing function of the inflow $X_t$, then by induction, at node $i$, whenever the SS, the MS, or the SM case is encountered, at equilibrium, the price of each active trade $p_{ki}$ on arc $ki$ is the same. This price can be rewritten as $p_i$, which is an affine decreasing function of the inflow $X_i$. When the source $s$ is reached, the price at $s$ satisfies $p_s = a_s = a_t - b_s X_s$. Since $a_t > a_s$ and $b_s > 0$, $X_s > 0$. At equilibrium, by the flow conservation property, the fact that goods are distributed accordingly to the positive convex coefficients in the SM case, and the assumption that there are no shortcuts, all arcs in $E$ are active, each seller $k \in S(i)$ offers $i$ the same price $p_i$, and $p_i = a_t - b_i X_i - \sum_{l \in C_P(i)} b_l X_l$.

For completeness, we list the closed-form of $b_i$ and the convex coefficient $\alpha_j$. The closed-form expression is used in Algorithm~\ref{alg: 1}.

\textbf{SS:} By \ref{ss-price},
\[b_i=2 b_j + \sum_{l \in C_P(j)} b_l.\]

\textbf{MS:} By \ref{ms-price},
\[b_i=b_j + \sum_{l \in C_P(j)} b_l.\]

\textbf{Simple SM:} By \ref{simple-sm-price} and \ref{app: eq: alpha_j},
\[b_i=\frac{2}{\sum_{j' \in B(i)}{\frac{1}{b_{j'}}}}+ 2\sum_{l \in C_S(i) \setminus C_P(i)}{b_l}+\sum_{l \in C_P(i)}{b_l} \text{ and } \alpha_j = \frac{\frac{1}{b_j}}{\sum_{j' \in B(i)}{\frac{1}{b_{j'}}}}.\]

\textbf{General SM:} By \ref{general-sm-price} and \ref{app: eq: prodcoeff},
\begin{equation} \label{general-sm} \tag{General SM}
b_i=2c_n(i)+\sum_{l \in C_P(i)}{b_l} \text{ and }
\alpha_j = \beta_{j}\prod_{p \mid h_p \in C_T(i,j) \setminus \{h_n\}}{\beta_p(i)}
\end{equation}
where $c_n(i)$, $\beta_p(i)$, and $\beta_j$ are defined by \eqref{app: eq: agg}, \eqref{app: eq: agg-n}, \eqref{app: eq: aggcoeff}, and \eqref{app: eq: vercoeff}.
\QED
\end{proof}


\subsection{Proof of Lemma~\ref{lem: 3.2}} \label{app: lem: 3.2}
\lemlcp*
\begin{proof}

We recall the feasibility problem \ref{eq: lcp}
\[
\begin{cases}
\sum_{j \in B(i)} \frac{\partial \Pi_i}{\partial x_{ij}} x_{ij}  = 0,\\
\frac{\partial \Pi_i}{\partial x_{ij}}  \leqslant 0  \quad \forall j \in B(i),\\
x_{ij} \geqslant 0 \quad \forall j \in B(i).
\end{cases}
\]
and the optimization problem \ref{eq: cp}
\[
\begin{aligned}
& \minimize_{x, X} & & \sum_{j \in B(i)} b_j x_{ij}^2 + \sum_{l \in C_S(i) \backslash C_P(i)} b_l X_l^2 \\
& \text{subject to}
& & a_t - 2 b_j x_{ij} - \sum_{l \in C_T(i,j)} 2 b_l X_l  \leqslant const_i & \forall j \in B(i), \\
& & & x_{ij} \geqslant 0 & \forall j \in B(i).
\end{aligned}
\]

	Consider the Lagrangian function:
	\[
	\begin{aligned}
	L(x, X, \lambda) &= \sum_{j \in B(i)} b_j x_{ij}^2 + \sum_{l \in C_S(i) \backslash C_P(i)} b_k X_l^2 \\
	&\quad - \sum_{j \in B(i)} \lambda_{ij} (a_t - 2 b_j x_{ij} - \sum_{l \in C_T(i,j)} 2 b_l X_l - const_i).
	\end{aligned}
	\]

	\textbf{Stationarity condition:}

\begin{itemize}
	\item By taking the derivative of $L$ with respect to $x_{ij}$, we have
	\[
	\frac{\partial L(x, X, \lambda)}{\partial x_{ij}} 
	= 2 b_j x_{ij} - 2 b_j \lambda_{ij} = 0
	\]
	which infers $x_{ij} = \lambda_{ij}$. 

	\item By taking the derivative of $L$ with respect to $X_l$ where $l \in C_S(i) \setminus C_P(i)$, we have
	\[
	\frac{\partial L(x, X, \lambda)}{\partial X_l}  
	= 2b_l X_l - \sum_{j \mid l \in C_P(j)} 2 b_l \lambda_{ij} = 0
	\]
	which infers $X_l = \sum_{j \mid l \in C_P(j)} \lambda_{ij} = \sum_{j \mid l \in C_P(j)} x_{ij}$. This is exactly the definition of $X_l$ (the total flow through $l$).
	\end{itemize}
	
	\textbf{Complementarity condition:}\\
	
	$\forall j \in B(s)$ (we recall that $x_{ij} = \lambda_{ij}$):
	\begin{align*}
	\lambda_{ij}(a_t - 2 b_j x_{ij} - \sum_{l \in C_T(i,j)} 2 b_l X_l - const_i) 
	= x_{ij} \frac{\partial \Pi_i}{\partial x_{ij}} 
	= 0.
	\end{align*}
	
	Combined with the primal feasibility conditions $\frac{\partial \Pi_i}{\partial x_{ij}}  \leqslant 0$ and $x_{ij} \geqslant 0$, the KKT condition of \ref{eq: cp} is equivalent to \ref{eq: lcp}. \ref{eq: cp} is strictly convex, so the solution is unique.
	\QED
\end{proof}

\subsection{Proof of Lemma~\ref{lem:shortcut}} \label{app:lem:shortcut}
\lemshortcut*
\begin{proof}
		From the structure of SPG and the flow conservation property at equilibrium, if path $l_{ij}$ has an active arc, then there exists a path from $i$ to $j$ where all arcs are active. To prove by contradiction, suppose $ij$ is a shortcut of path $l_{ij} = (i, v_1, ..., v_k, j)$, without loss of generality, we can assume that all arcs in the path $l_{ij}$ are active.

    Since firms never sell goods at a lower price than the buying price, by Observation~\ref{obs:inc-p},
    \begin{align*}
    p_{i v_1} \leqslant p_{v_1 v_2} \leqslant \dots \leqslant p_{v_{k-1} v_k} \leqslant p_{v_k j} = p_j.
    \end{align*}

    Consider the case that $p_{i v_1} < p_j$ at the equilibrium, by the structure of SPG and the flow conservation property, all the flow from $i$ to $v_1$ will go to firm $j$. If firm $i$ moves $x_{i v_1}$ amount of flow from $i v_1$ to $ij$, the total flow through $j$ will be the same (since $p_j$ is a function of $X_j$), and $p_j$ will remain the same price. Therefore, firm $i$ is better off by the difference from the selling revenue
    \begin{align*}
    p_j (x_{ij} + x_{i v_1}) - p_j x_{ij} - p_{v_1} x_{i v_1} > 0,
    \end{align*}
    which cannot happen at an equilibrium. Thus, $p_{i v_1} = p_j$ must hold, and
    \begin{align*}
    p_{iv_1} =  p_{v_1 v_2} = \dots = p_{v_{k-1}v_k} = p_{v_k j} = p_j.
    \end{align*}

    Now consider the optimal decision for $v_k$, if she buys all the goods offered to her and sell them to $j$, her profit is $0$, because $p_{v_{k-1}v_k} = p_j$. However, she would make a positive profit if she accepts and offers $j$ less amount of goods. Because this would decrease the flow to $j$ and raise the optimal price of $j$ from $p_j$ to $p_j'$, that is,
    \begin{align*}
    p_j' > p_j = p_{v_{k-1}v_k},
    \end{align*}
    which contradicts to the flow conservation property at equilibrium. Hence, the path $l_{ij}$ is inactive. \QED



\end{proof}

\section{Proofs in Section \ref{sec: eqp}}

\subsection{Proof of Lemma~\ref{lem: 4.3}} \label{app: lem: 4.3}

\lemucf*

\begin{proof} The proof is done case by case. Suppose $i$ is the seller of a trade:

\begin{itemize}
\item For the SS case, $X_i = X_j = x_{ij}$ and $C_P(i) = C_P(j)$. Consider the utility of $i$, by equation~\ref{eq: SS}:
	\begin{align*}
	\Pi_i &= (p_j - p_i) x_{ij}  \\
	&= (b_{i} X_i - b_j X_j) x_{ij} \\
	&= (b_{i} - \frac{b_i - \sum_{l \in C_P(i)} b_l}{2}) X_i^2 \\
	&= \frac{1}{2}(b_{i} + \sum_{l \in C_P(i)} b_l) X_i^2.
	\end{align*}

\item For the SM case, for each $j \in B(i)$, we recall that $ij$ is active and $p_{ki} = p_i$ at equilibrium so the derivative in \eqref{sm-dev-0} must be 0:
\[\frac{\partial \Pi_i}{\partial x_{ki}} = a_t - (2c_n(i)+\sum_{l \in C_P(i)}{b_l})X_i - \sum_{l \in C_P(i)}{b_l X_l} - p_i = 0.\]
We recall the price of $j$ in \eqref{sm-pj}:
\[p_j= a_t - c_n(i) X_i - \sum_{l \in C_P(i)}{b_l X_l}.\]
By rearranging and \ref{general-sm}, we have
\begin{align*}
p_j - p_i &= (c_n(i)+\sum_{l \in C_P(i)}b_l)X_i \label{sm-u-cn} \\
&= \frac{1}{2}(b_{i} + \sum_{l \in C_P(i)} b_l)X_i.
\end{align*}
Therefore,
	\begin{align*}
	\Pi_i &= \sum_{j \in B(i)} (p_j - p_i)x_{ij} \\
	&= \sum_{j \in B(i)} \frac{1}{2}(b_{i} + \sum_{l \in C_P(i)} b_l)X_i x_{ij} \\
	&= \frac{1}{2}(b_{i} + \sum_{l \in C_P(i)} b_l)X^2_i.
	\end{align*}

\item For the MS case, $x_{ij} = X_i$ and $C_P(i) = C_P(j) \sqcup \{j\}$. Consider the utility of $i$:
\begin{align*}
\Pi_i &= (p_j - p_i)x_{ij} \\
&= (b_i X_i + \sum_{l \in C_P(j)} b_l X_l + b_j X_j - b_j X_j - \sum_{l \in C_P(j)} b_l X_l)X_i \\
&= b_i X^2_i.
\end{align*}
By equation~\ref{eq: MS}:
	\begin{align*}
	b_i &= b_j + \sum_{l \in C_P(j)} b_l \\
	&= \sum_{l \in C_P(i)} b_l \\
	&= \frac{b_i + \sum_{l \in C_P(i)} b_l}{2}.
	\end{align*}
Therefore,
\[\Pi_i = b_i X^2_i = \frac{1}{2}(b_{i} + \sum_{l \in C_P(i)} b_l)X^2_i.\]
\end{itemize}
\QED
\end{proof}

\subsection{Proof of Proposition~\ref{cor: 4.1}} \label{app: cor: 4.1}

\proptwousssm*

\begin{proof}
Consider the arc $ij \in E$. For the SS case, $X_i = X_j = x_{ij}$ and $C_P(i) = C_P(j)$. By Lemma~\ref{lem: 4.3} and \ref{eq: SS},
\begin{align*}
\Pi_i &= \frac{1}{2}(b_{i} + \sum_{l \in C_P(i)} b_l)X^2_i \\
&= \frac{1}{2}(2b_{j} + 2\sum_{l \in C_P(j)} b_l)X^2_i \\
&= 2\Pi_j.
\end{align*}

For the SM case, $X_j = x_{ij}$, by \eqref{cnx} and Lemma \ref{lem: 4.3},
\begin{align*}
c_n(i)X^2_i &= c_n(i)X_i \sum_{j \in B(i)}{x_{ij}} \\
&= \sum_{j \in B(i)}{x_{ij}(b_j x_{ij} + \sum_{l \in C_T(i,j)}{b_l X_l})} \\
&\geqslant \sum_{j \in B(i)}{(b_j x^2_{ij} + \sum_{l \in C_T(i,j)}{b_l} x^2_{ij})} \\
&=2\sum_{j \in B(i)}{\Pi_j}.
\end{align*}
By \ref{general-sm} and Lemma \ref{lem: 4.3},
\begin{align*}
\Pi_i &= \frac{1}{2}(2c_n(i)X_i+2\sum_{l \in C_P(i)}b_lX_i)X_i \\
&\geqslant c_n(i)X^2_i\\
&\geqslant 2\sum_{j \in B(i)}{\Pi_j}.
\end{align*}
\QED
\end{proof}

\subsection{Proof of Proposition~\ref{thm: 4.4}} \label{app: thm: 4.4}

\proptwou*

\begin{proof}
	Without loss of generality, we consider the closest dominating parent $i$ of $j$. Suppose $j$ is the buyer of a trade.

In the SS or SM case, the closest dominating parent of $j$ is $i \in S(j)$ so $ij \in E$. The claim follows by Proposition~\ref{cor: 4.1}.

Suppose along the path from $i$ to $j$, $j$ ends up to be a single buyer in the MS case. Then by \ref{general-sm} and Lemma \ref{lem: 4.3},
\begin{align*}
\Pi_i &= \frac{1}{2}(b_i+\sum_{l \in C_P(i)}b_l)X^2_i \\
&= \frac{1}{2}(2c_n(i)+2\sum_{l \in C_P(i)}b_l)X^2_i \\
&\geqslant (b_j +\sum_{l \in C_P(j)}b_l)X^2_i \\
&\geqslant (b_j+\sum_{l \in C_P(j)}b_l)X^2_j \\
&= 2\Pi_j
\end{align*}
where the first inequality holds by $c_n(i) \geqslant b_j$, which can be proved by induction and equation \ref{app: eq: agg}, \ref{app: eq: agg-n}, \ref{app: eq: aggcoeff}, and \ref{app: eq: vercoeff}, and the fact that $i$ is the closest dominating parent of $j$ implies $C_P(j) \subseteq C_P(i)$; the second inequality holds by $X_j \leqslant X_i$. \QED
\end{proof}

\subsection{Proof of Lemma~\ref{lem: 4.1}} \label{app: lem: 4.1}
\lemcp*
\begin{proof}
	By Theorem~\ref{lem:pq-relation}, $p_s = a_t - b_s X_s$. While calculating the price function from the sink, one can show that by induction and \ref{eq: SS}, \ref{eq: MS}, and \ref{general-sm}, $b_i$ where $i \in V \cup \{s\}$ changes proportionally to $b_t$, so $\lambda(Y)$ is a constant.

	The remaining is to show that $b_s \geqslant 2b_t$ by induction. For $ij \in E$, in the SS case, by \ref{eq: SS}, $b_i \geqslant 2 b_j$; in the MS case, by \ref{eq: MS}, $b_i \geqslant 2 b_j$; in the SM case, $b_i \geqslant 2 c_n(i) \geqslant 2 b_{h_n}$ by induction and equation \ref{app: eq: agg}, \ref{app: eq: agg-n}, \ref{app: eq: aggcoeff}, \ref{app: eq: vercoeff}, and \ref{general-sm}, where $h_n$ is the farthest node from $i$ in $C_S(i)$. Combining these cases, $b_s \geqslant 2b_t$.
\QED
\end{proof}

\subsection{Proof of Proposition~\ref{prop: 4.1}} \label{app: prop: 4.1}
\propinc*
\begin{proof}
	It follows that $X_s = \frac{a_t - a_s}{b_s}$, so the increasing demand at market ($a_t$) or decreasing cost at the source ($a_s$) will make the flow value larger. Since $b_t$ is not changed, by \eqref{eq: sw1}, the coefficients of the quadratic terms ($X_i^2$ and $X_t^2$) do not change either. By lemma~\ref{lem: 3.3}, the flow is distributed proportionally to the convex coefficients in the SM case (for the SS and MS case, just sum the flow from the upstream), so the flow increases proportionally as well. Therefore, the flow and welfare efficiency both increases, and by Lemma~\ref{lem: 4.3}, the utility of each individual firm increases.
\QED
\end{proof}

\subsection{Proof of Lemma~\ref{lem: 4.z}} \label{app: lem: 4.z}
\lempc*

\begin{proof}
The parallel composition $P(Y,Z)$ creates an SM case at the source. Therefore, it suffices to find the convex coefficients for solving \ref{ls-sm}. Let $\alpha^Y_j$ (respectively $\alpha^Z_j$) be the convex coefficient for $s_Y j \in E(Y)$ (respectively $s_Z j \in E(Z)$) where $j \in B(s_Y)$ (respectively $B(s_Z)$). If $s_Y$ (respectively $s_Z$) is the buyer of the SS case, then $\alpha^Y_j=1$ (respectively $\alpha^Z_j=1$). The convex coefficients are $\frac{\lambda(Z)-2}{\lambda(Y)+\lambda(Z)-2}\alpha^Y_j$ for each $j \in B(s_Y)$ and $\frac{\lambda(Y)-2}{\lambda(Y)+\lambda(Z)-2}\alpha^Z_j$ for each $j \in B(s_Z)$. Suppose $b_{t_{P(Y,Z)}}=b_{t_Y}=b_{t_Z}$, $b_{s_Y}=\lambda(Y)b_{t_Y}$, and $b_{s_Z}=\lambda(Z)b_{t_Z}$, then by \eqref{app: eq: agg-n} and \ref{general-sm},
\[b_{s_{P(Y,Z)}} = (\frac{(\lambda(Y)-2)(\lambda(Z)-2)}{\lambda(Y)+\lambda(Z)-4} + 2 )b_{t_{P(Y,Z)}}.\]
\QED
\end{proof}

\section{Proofs in Section \ref{sec: ext}}

\subsection{Proof Sketch of Proposition~\ref{prop: 4.z}} \label{app: prop: 4.z}

\propsmspg*

\noindent{\bf Proof Sketch.} Without loss of generality, we consider shortcut-free SMSPG. The derivation of the equilibrium price is similar to the proof of Theorem~\ref{thm: 3.1}. We inductively start from the sink markets in $T$ and the price function at each firm is affine decreasing before an SM case is reached. For the SM case, the buyer does not have the control over the buying cost and inflow at equilibrium, so a convex quadratic program \ref{eq: dist} can be derived. By considering the Lagrangian function, the linear system \ref{ls-sm} can be formulated. The solution of \ref{ls-sm} is to distribute the flow proportionally to the convex coefficients where each of them is positive, so each trade is active. The flow distribution according to the convex coefficients gives a closed-form expression of the price offered to the buyer. This procedure goes on inductively until the source is reached. When the source is reached, we compute the total flow value of the network and use Algorithm~\ref{alg: 2} to compute the equilibrium flow. The uniqueness of the equilibrium follows by Lemma~\ref{lem: 3.2}. The problem of distributing the flow for an SM case buyer can be described as an \ref{eq: lcp} which has an equivalent \ref{eq: cp}, and its solution is unique.

We note that when all markets have the same demand, it suffices to focus on the calculation of $b_i$ for each $i \in V \cup \{s\}$. The equilibrium calculation is more complicated if that is not the case.

\newpage

\subsection{Proof of Remark~\ref{rm1}} \label{app: rm1}

We consider the following supply chain network:

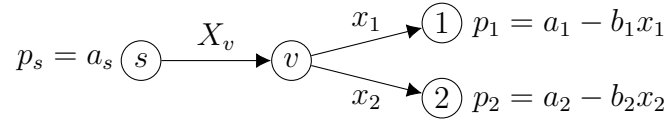
\begin{center}
\begin{tikzpicture}
	\node[vertex] (b) at (0,0) {$s$}; 
	\node at (-1, 0) {$p_s=a_s$};
	\node[vertex] (a) at (2,0) {$v$}; 
	\node[vertex] (t1) at (4,.5) {$1$}; 
	\node[vertex] (t2) at (4,-.5) {$2$}; 
	\node at (5.7, .5) {$p_1 = a_1 - b_1 x_1$};
	\node at (5.7, -.5) {$p_2 = a_2 - b_2 x_2$};
	\path[->]
		(b) edge node [above] {$X_v$} (a) 
		(a) edge node [above] {$x_1$} (t1) 
		(a) edge node [below] {$x_2$} (t2);
\end{tikzpicture}
\end{center}

For simplicity, we denote the first market price as $p_1$ and the second market price as $p_2$. The production cost is a constant $a_s$. Let the inflow of market 1 be $x_1$ and the inflow of market 2 be $x_2$. Suppose the two price functions at the markets are:
\begin{align*}
p_1 = a_1 - b_1 x_1, \\
p_2 = a_2 - b_2 x_2,
\end{align*}
where $a_1 \geqslant a_2 \geqslant a_s$.

Throughout the proof, we add a superscript $h$ for variables under the high price strategy and $l$ for the low price strategy.

\begin{itemize}
\item The low price strategy always gives a higher flow value than the high price strategy.

\begin{proof}

With the high price strategy, $x^h_2 = 0$, so it is equivalent to regard the entire supply chain as a line graph from $s$ to $v$ then from $v$ to market $1$. Therefore, $p^h_v = a_1 - 2b_1 X^h_v$. The optimal flow $X^h_v$ under the high price strategy is
\begin{align}
a_s = a_1 - 4 b_1 X^h_v \implies X^h_v = \frac{a_1 - a_s}{4 b_1}. \label{Xhv}
\end{align}

Under the low price strategy, the utility of $v$ is
\[\Pi_v = (a_1 - b_1 x_1) x_1 + (a_2 - b_2 x_2)x_2 - p^l_v(x_1 + x_2).\]

$x_1 > 0$ and $x_2 > 0$, so $\frac{\partial \Pi_v}{\partial x_1} = 0$ and $\frac{\partial \Pi_v}{\partial x_2} = 0$:

\begin{equation} \label{app: flow_rel}
\begin{cases}
a_1 - 2b_1x_1 - p^l_v = 0,\\
a_2 - 2b_2x_2 - p^l_v = 0,\\
\end{cases}
\implies
a_1 b_2 + a_2 b_1 - 2b_1 b_2 (x_1+x_2) - (b_1+b_2)p^l_v = 0.
\end{equation}

The optimal flow $X^l_v$ under the low price strategy is
\begin{align}
 p^l_v &= (a_1/b_1 + a_2/b_2)B - 2B X^l_v, \nonumber \\
 p_s &= a_s = (a_1/b_1 + a_2/b_2)B - 4B X^l_v, \nonumber \\
 X^l_v &= \frac{(a_1/b_1 + a_2/b_2)B - a_s}{4 B}, \label{app: flow_val}
\end{align}
where $B = \frac{1}{\frac{1}{b_1} + \frac{1}{b_2}}$.

Then we have the difference of total flow between these two strategies:
\begin{align}
X^l_v - X^h_v &= \frac{(a_1/b_1 + a_2/b_2)B - a_s}{4 B} - \frac{a_1 - a_s}{4 b_1} \nonumber \\
&= \frac{a_2}{4 b_2} - \frac{a_s}{4 B} + \frac{a_s}{4 b_1} \nonumber \\
&= \frac{a_2}{4 b_2} - \frac{a_s}{4 b_1} - \frac{a_s}{4 b_2} + \frac{a_s}{4 b_1} \nonumber \\
&= \frac{a_2}{4 b_2} - \frac{a_s}{4 b_2} \nonumber \\
&\geqslant 0. \label{flow_diff}
\end{align}
\QED
\end{proof}

\item When the demand difference between two markets is small enough, the low price strategy gives better utility for the source. When the difference is large enough, the high price strategy gives better utility for the source.

\begin{proof}
	Let $\Pi_v$ be the utility of firm $v$, and $\Pi_s$ be the utility of firm $s$. Under the high price strategy:
	\begin{align}
	 \Pi^h_v &= b_1{X^h_v}^2, \label{Pihv}\\
	 \Pi^h_s &= 2b_1{X^h_v}^2 = \frac{(a_1 - a_s)^2}{8b_1}. \label{Pihs}
	\end{align}

	To get the social welfare under the low price strategy, from \eqref{app: flow_rel} and \eqref{app: flow_val}:
	\begin{align*}
	& p^l_v = a_1 - 2b_1 x_1 = a_2 - 2b_2 x_2, \\
	& x_1 + x_2 = X^l_v,
	\end{align*}
	infers
	\begin{align*}
	& x_1 = \frac{a_1 - a_2 + 2b_2 X^l_v}{2b_1 + 2b_2}, \\
	& x_2 = \frac{ 2b_1 X^l_v - a_1 + a_2}{2b_1 + 2b_2}, \\
	\end{align*}
where 

\[X^l_v = \frac{(a_1/b_1 + a_2/b_2)B - a_s}{4 B} \text{ and } B = \frac{1}{\frac{1}{b_1} + \frac{1}{b_2}}.\]
	
	We have
	\begin{align}
	 \Pi^l_s &= (p^l_v - a_s)X^l_v =  2B{X^l_v}^2 \label{Pils} \\
	&= 2B [\frac{(a_1/b_1 + a_2/b_2)B - a_s}{4 B}]^2 = \frac{[(a_1/b_1 + a_2/b_2)B - a_s]^2}{8 B}. \nonumber
	\end{align}

By taking the ratio between \eqref{Pihs} and \eqref{Pils},
	\begin{align*}
		\frac{\Pi^h_s}{\Pi^l_s} = \frac{b_1 {X^h_v}^2}{B{X^l_v}^2} = \frac{b_1 + b_2}{b_2}\frac{{X^h_v}^2}{(X^h_v+\Delta)^2}
	\end{align*}
where $\Delta = \frac{a_2-a_s}{4 b_2}$ is irrelevant to $a_1$ by \eqref{flow_diff}.

$s$ has no preference between the high price and the low price strategy when
\begin{align*}
&\quad \quad \frac{X^h_v}{X^h_v+\Delta} = \sqrt{\frac{b_1}{b_1+b_2}} \\
&\implies \sqrt{b_1+b_2}(\frac{a_1 - a_s}{4 b_1}) = \sqrt{b_1}(\frac{a_1 - a_s}{4 b_1}+\Delta) \\
&\implies \frac{\sqrt{b_1+b_2} - \sqrt{b_1}}{4b_1}a_1 = \frac{\sqrt{b_1+b_2} - \sqrt{b_1}}{4b_1}a_s + \sqrt{b_1}\Delta \\
&\implies a_1 = a_s + \frac{4b_1^{\frac{3}{2}}\Delta}{\sqrt{b_1+b_2} - \sqrt{b_1}}. \\
\end{align*}

$X^h_v$ increases linearly to $a_1$. When $a_1$ is below this value, then the ratio is smaller than 1 and $s$ prefers the low price strategy. When $a_1$ is above this value, then the ratio is greater than 1 and $s$ prefers the high price strategy.
\QED

\end{proof}

\item The low price strategy always produces higher social welfare.
\begin{proof}

Let $CS$ be the consumer surplus, and $SW$ be the social welfare. Under the high price strategy, by \eqref{Xhv}, \eqref{Pihv}, and \eqref{Pihs}:
	\begin{align*}
	 CS^h &= \frac{1}{2}b_1{X^h_v}^2, \\
SW^h &= CS^h + \Pi^h_v + \Pi^h_s = \frac{7}{2} b_1 {X^h_v}^2 = \frac{7}{2} b_1 (\frac{a_1 - a_s}{4b_1})^2 = \frac{7(a_1 - a_s)^2}{32b_1}.
	\end{align*}

Under the low price strategy:
\begin{align*}
	 	CS^l &= \frac{1}{2}b_1x_1^2 + \frac{1}{2}b_2x_2^2, \\
	 \Pi^l_v &= x_1(p^l_1-p^l_v) + x_2(p^l_2-p^l_v) = b_1x_1^2 + b_2x_2^2, \\
		SW^l &= CS^l + \Pi^l_v + \Pi^l_s \\
		&= \frac{3}{2}(b_1 x_1^2 + b_2 x_2^2) + \Pi^l_s\\
		&= \frac{3[b_1(a_1 - a_2 + 2b_2 X^l_v)^2 + b_2(2b_1 X^l_v - a_1 + a_2)^2]}{8(b_1 + b_2)^2} + \Pi^l_s\\
		&= \frac{3b_1[(a_1-a_2)^2 + 4b_2^2{X^l_v}^2 + 4a_1b_2X^l_v-4a_2b_2X^l_v]}{8(b_1 + b_2)^2} \\
		&\quad +\frac{3b_2[(a_1-a_2)^2 + 4b_1^2{X^l_v}^2 - 4a_1b_1X^l_v+4a_2b_1X^l_v]}{8(b_1 + b_2)^2} + \Pi^l_s \\
		&= \frac{3[(b_1+b_2)(a_1-a_2)^2+4b_1b_2(b_1+b_2){X^l_v}^2]}{8(b_1 + b_2)^2} + \Pi^l_s \\
		&= \frac{3(a_1-a_2)^2}{8(b_1+b_2)} + \frac{3B{X^l_v}^2}{2} + 2B{X^l_v}^2 \\
		&= \frac{3(a_1-a_2)^2}{8(b_1+b_2)} + \frac{7B{X^l_v}^2}{2}.
	\end{align*}

By taking the difference and \eqref{Pils},
\begin{align*}
		SW^l - SW^h &= \frac{3(a_1-a_2)^2}{8(b_1+b_2)} + \frac{7}{2}[\frac{b_1 b_2}{b_1 + b_2}(X^h_v+\Delta)^2 - b_1{X^h_v}^2] \\
		&= \frac{3(a_1-a_2)^2}{8(b_1+b_2)} + \frac{7}{2}[\frac{-b_1^2}{b_1 + b_2}{X^h_v}^2+\frac{b_1(a_2-a_s)X^h_v}{2(b_1+b_2)}+\frac{b_1(a_2-a_s)^2}{16b_2(b_1+b_2)}] \\
		&= \frac{12b_2(a_1-a_2)^2-7b_2(a_1-a_s)^2+14b_2(a_2-a_s)(a_1-a_s)+7b_1(a_2-a_s)^2}{32b_2(b_1+b_2)} \\
		&= \frac{12b_2(a_1-a_2)^2-7b_2[(a_1-a_s)-(a_2-a_s)]^2+7(b_1+b_2)(a_2-a_s)^2}{32b_2(b_1+b_2)} \\
		&= \frac{5b_2(a_1-a_2)^2+7(b_1+b_2)(a_2-a_s)^2}{32b_2(b_1+b_2)} \\
		&\geqslant 0
	\end{align*}
where $\Delta = \frac{a_2-a_s}{4 b_2}$. \QED

\end{proof}
\end{itemize}


\section{Examples}

\begin{example}[Merging Child Nodes] \label{app:ex:mc}
\mbox{}\\

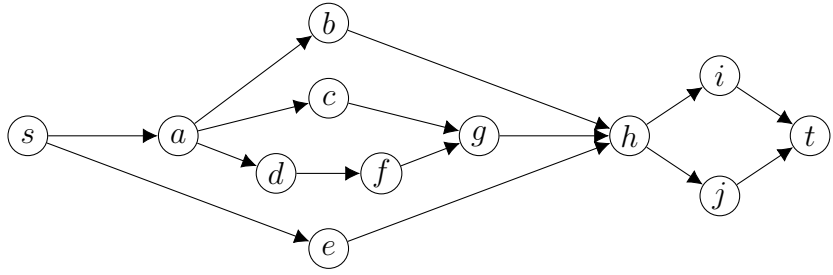
\begin{center}
\begin{tikzpicture}
	\node[vertex] (s) at (-2,0) {$s$}; 
	\node[vertex] (a) at (0,0) {$a$}; 
	\node[vertex] (b) at (2,1.5) {$b$}; 
	\node[vertex] (c) at (2,0.5) {$c$}; 
	\node[vertex] (d) at (1.3,-0.5) {$d$};
	\node[vertex] (e) at (2,-1.5) {$e$};
	\node[vertex] (f) at (2.7,-.5) {$f$};
	\node[vertex] (g) at (4,0) {$g$};
  	\node[vertex] (h) at (6,0) {$h$};
  	\node[vertex] (i) at (7.2,.8) {$i$};
  	\node[vertex] (j) at (7.2,-.8) {$j$};
  	\node[vertex] (t) at (8.4,0) {$t$};
	\path[->]
		(s) edge (a)
		(s) edge (e)
		(a) edge (b)
		(a) edge (c)
		(a) edge (d)
		(c) edge (g)
		(d) edge (f)
		(f) edge (g)
		(b) edge (h)
		(e) edge (h)
		(g) edge (h)
		(h) edge (i)
		(h) edge (j)
		(i) edge (t)
		(j) edge (t);
\end{tikzpicture}
\end{center}

In this graph, for node $a$, $C_S(a) = \{g, h\}$, because $\{g, h\} \subseteq C(a)$ and there are multiple disjoint paths from $a$ to $g$ and $h$, while $t \notin C_S(a)$ because all the paths from $a$ to $t$ must go through the common node $h$; $C_P(a) = \{h\}$ because $h \in C(a)$, $s \in P(a)$, and there are multiple disjoint paths from $s$ to $h$; $C_T(a,b) = \emptyset$, while $C_T(a,c) = \{g\}$.

For node $c$, $C_P(c) = \{g, h\}$, while $C_S(c) = \emptyset$; For node $g$, $C_P(g) = \{h\}$, while $C_S(g) = \emptyset$.

By Observation~\ref{obs:mc}, since $a$ and $c$ satisfy the SM relation, $C_P(c) = \{g, h\} = C_P(a) \sqcup C_T(a,c)$. $c$ and $g$ satisfy the MS relation, so $C_P(c) = \{g, h\} = C_P(g) \sqcup \{g\}$.
\end{example}

\begin{example}[Price Function Computation by Algorithm~\ref{alg: 1}] \label{app: ex: price_compute_alg1}
\mbox{}\\

Consider the following network.

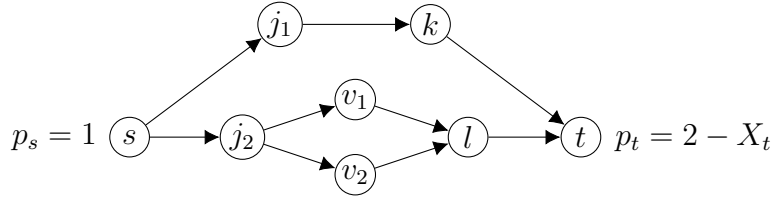
\begin{center}
\begin{tikzpicture}[baseline=0]
	\node at (-1, 0) {$p_{s} = 1$};
	\node[vertex] (s) at (0,0) {$s$}; 
	\node[vertex] (a) at (2,1.5) {$j_1$}; 
	\node[vertex] (a1) at (1.5,0) {$j_2$};
	\node[vertex] (b) at (3,0.5) {$v_1$}; 
	\node[vertex] (c) at (3,-0.5) {$v_2$};
	\node[vertex] (d) at (4.5,0) {$l$};
	\node[vertex] (f) at (4,1.5) {$k$};
  \node[vertex] (t) at (6,0) {$t$};
	\node at (7.5, 0) {$p_{t} = 2 - X_t$};
	\path[->]
		(s) edge (a)
		(s) edge (a1) 
		(a1) edge (b)
		(a1) edge (c)
		(a) edge (f)
		(b) edge (d)
		(c) edge (d)
		(f) edge (t)
		(d) edge (t)
;
\end{tikzpicture}
\end{center}
We recall the equations in Algorithm~\ref{alg: 1}:\\

\textbf{\ref{eq: SS}:}
\begin{align*}
b_i = 2 b_j + \sum_{l \in C_P(j)} b_l.
\end{align*}

\textbf{\ref{eq: MS}:}
\begin{align*}
b_i = b_j + \sum_{l \in C_P(j)} b_l.
\end{align*}

\textbf{\ref{eq: SM}:}
\begin{align*}
	&b_i = \frac{2}{\sum_{j \in B(i)} \frac{1}{b_{j}}} + 2 \sum_{l \in C_S(i) \setminus C_P(i)} b_l + \sum_{l \in C_P(i)} b_l, \\
	&\alpha_j = \frac{\frac{1}{b_{j}}}{\sum_{j' \in B(i)} \frac{1}{b_{j'}}} \text{ for each $j \in B(i)$}.
	\end{align*}

By Algorithm~\ref{alg: 1}, for the MS case from $k$ and $l$ to $t$:
\begin{align*}
p_k = 2 - X_k - X_t, \\
p_l = 2 - X_l - X_t.
\end{align*}

For the MS case from $v_1$ and $v_2$ to $l$:
\begin{align*}
p_{v_2} = 2 - 2X_{v_1} - X_l - X_t, \\
p_{v_3} = 2 - 2X_{v_2} - X_l - X_t.
\end{align*}

For the SS case from $j_1$ to $k$:
\begin{align*}
p_{j_1} = 2 - 3X_{j_1} - X_t.
\end{align*}

For the SM case from $j_2$ to $v_1$ and $v_2$:
\begin{align*}
p_{j_2} &= 2 - (\frac{2}{\frac{1}{2}+\frac{1}{2}} + 2 + 1)X_{j_2} - X_t \\
&= 2 - 5X_{j_2} - X_t, \\
\alpha_{v_1} &= \alpha_{v_2} = \frac{\frac{1}{2}}{\frac{1}{2}+\frac{1}{2}} = \frac{1}{2}.
\end{align*}

For the SM case from $s$ to $j_1$ and $j_2$:
\begin{align*}
p_{s} &= 2 - (\frac{2}{\frac{1}{3}+\frac{1}{5}} + 2)X_t \\
&= 2 - \frac{23}{4} X_t, \\
X_s &= \frac{4}{23}(2 - 1) = \frac{4}{23}, \\
\alpha_{j_1} &= \frac{\frac{1}{3}}{\frac{1}{3}+\frac{1}{5}} = \frac{5}{8}, \\
\alpha_{j_2} &= \frac{\frac{1}{5}}{\frac{1}{3}+\frac{1}{5}} = \frac{3}{8}.
\end{align*}

\end{example}

\begin{example}[Price Function Computation for General SM] \label{app: ex: price_compute}
\mbox{}\\

Consider the following network.

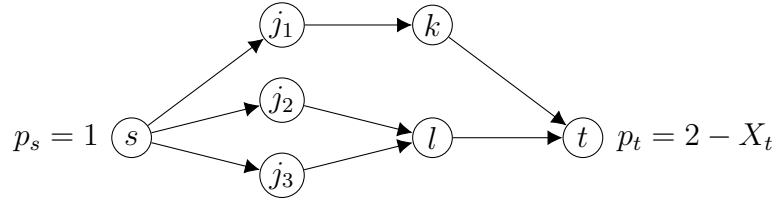
\begin{center}
\begin{tikzpicture}[baseline=0]
	\node at (-1, 0) {$p_{s} = 1$};
	\node[vertex] (s) at (0,0) {$s$}; 
	\node[vertex] (a) at (2,1.5) {$j_1$}; 
	\node[vertex] (b) at (2,0.5) {$j_2$}; 
	\node[vertex] (c) at (2,-0.5) {$j_3$};
	\node[vertex] (d) at (4,0) {$l$};
	\node[vertex] (f) at (4,1.5) {$k$};
  \node[vertex] (t) at (6,0) {$t$};
	\node at (7.5, 0) {$p_{t} = 2 - X_t$};
	\path[->]
		(s) edge (a) 
		(s) edge (b)
		(s) edge (c)
		(a) edge (f)
		(b) edge (d)
		(c) edge (d)
		(f) edge (t)
		(d) edge (t)
;
\end{tikzpicture}
\end{center}
From the equations in \ref{app:lem:pq-relation}:\\

\textbf{SS:}
\[b_i=2 b_j + \sum_{l \in C_P(j)} b_l.\]

\textbf{MS:}
\[b_i=b_j + \sum_{l \in C_P(j)} b_l.\]

\textbf{General SM:}
\[b_i=2c_n(i)+\sum_{l \in C_P(i)}{b_l} \text{ and } \alpha_j = \beta_{j}\prod_{p \mid h_p \in C_T(i,j) \setminus \{h_n\}}{\beta_p(i)}\]
where $c_n(i)$, $\beta_p(i)$, and $\beta_j$ are defined by equation \ref{app: eq: agg}, \ref{app: eq: agg-n}, \ref{app: eq: aggcoeff}, and \ref{app: eq: vercoeff}.

By the backward Algorithm~\ref{alg: 1}, for the MS case from $k$ and $l$ to $t$:
\begin{align*}
p_k = 2 - X_k - X_t, \\
p_l = 2 - X_l - X_t.
\end{align*}

For the MS case from $j_2$ and $j_3$ to $l$:
\begin{align*}
p_{j_2} = 2 - 2X_{j_2} - X_l - X_t, \\
p_{j_3} = 2 - 2X_{j_3} - X_l - X_t.
\end{align*}

For the SS case from $j_1$ to $k$:
\begin{align*}
p_{j_1} = 2 - 3X_{j_1} - X_t.
\end{align*}

The remaining is the general SM case from $s$ to $j_1$, $j_2$, and $j_3$. By following the notations in \ref{para:sm}, $C_S(s) = \{h_1, h_2\}$ where $h_1 = l$ and $h_2 = t$. $B_2(s) = \{j_1\}$, $B_1(s) = \{j_2, j_3\}$, $P_2(i) = \{h_1\}$, and $P_1(s) = \emptyset$. We start with the aggregate variable $c_1(s)$ since $P_1(s) = \emptyset$.

\[c_1(s) = \frac{1}{\frac{1}{b_{j_2}}+\frac{1}{b_{j_3}}} + b_l = \frac{1}{\frac{1}{2}+\frac{1}{2}} + 1 = 2,\]

\[b_s = \frac{2}{\frac{1}{b_{j_1}} + \frac{1}{c_1(s)}} + 2 b_t = \frac{2}{\frac{1}{3} + \frac{1}{2}} + 2 = \frac{22}{5}.\]

Therefore, $p_s = 1 - \frac{22}{5}X_s = 0$ and $X_s = \frac{5}{22}$. For the convex coefficients,

\[\beta_{j_2} = \beta_{j_3} = \frac{\frac{1}{2}}{\frac{1}{2}+\frac{1}{2}} = \frac{1}{2},\]

\[\alpha_{j_1} = \beta_{j_1} = \frac{\frac{1}{3}}{\frac{1}{3} + \frac{1}{2}} = \frac{2}{5},\]

\[\beta_1(s) = \frac{\frac{1}{2}}{\frac{1}{3} + \frac{1}{2}} = \frac{3}{5},\]

\[\alpha_{j_2} = \alpha_{j_3} = \frac{1}{2} \beta_{b_{B_1(s)}} = \frac{3}{10}.\]

\end{example}

\begin{example}[Price and Flow Computation by Algorithm~\ref{alg: 2}] \label{app: ex: flow_compute_alg2}
\mbox{}\\

We use the same SPG as in Example~\ref{app: ex: price_compute_alg1}. 

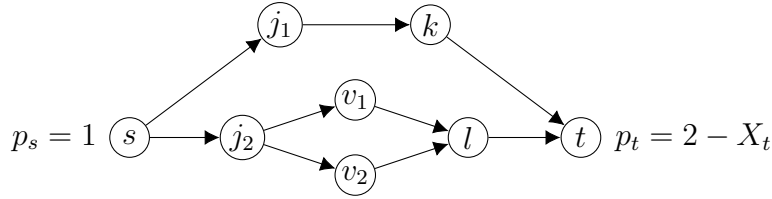
\begin{center}
\begin{tikzpicture}[baseline=0]
	\node at (-1, 0) {$p_{s} = 1$};
	\node[vertex] (s) at (0,0) {$s$}; 
	\node[vertex] (a) at (2,1.5) {$j_1$}; 
	\node[vertex] (a1) at (1.5,0) {$j_2$};
	\node[vertex] (b) at (3,0.5) {$v_1$}; 
	\node[vertex] (c) at (3,-0.5) {$v_2$};
	\node[vertex] (d) at (4.5,0) {$l$};
	\node[vertex] (f) at (4,1.5) {$k$};
  \node[vertex] (t) at (6,0) {$t$};
	\node at (7.5, 0) {$p_{t} = 2 - X_t$};
	\path[->]
		(s) edge (a)
		(s) edge (a1) 
		(a1) edge (b)
		(a1) edge (c)
		(a) edge (f)
		(b) edge (d)
		(c) edge (d)
		(f) edge (t)
		(d) edge (t)
;
\end{tikzpicture}
\end{center}

From Example~\ref{app: ex: price_compute_alg1}, $X_s = \frac{4}{23}$, $\alpha_{j_1} = \frac{5}{8}$, $\alpha_{j_2} = \frac{3}{8}$, and $\alpha_{v_1} = \alpha_{v_2} = \frac{1}{2}$. By Lemma~\ref{lem: 3.3},

\begin{align*}
x_{sj_1} &= \frac{4}{23} \times \frac{5}{8} = \frac{5}{46}, \\
x_{sj_2} &= \frac{4}{23} \times \frac{3}{8} = \frac{3}{46}, \\
\end{align*}
\begin{align*}
p_{j_1} &= 1 - 3X_{j_1} - X_t = 2 - 3 \times \frac{5}{46} - \frac{4}{23} = \frac{3}{2}, \\
p_{j_2} &= 1 - 5X_{j_2} - X_t = 2 - 5 \times \frac{3}{46} - \frac{4}{23} = \frac{3}{2},
\end{align*}
\[x_{j_2v_1} = x_{j_2v_2} = \frac{x_{sj_2}}{2} = \frac{3}{92},\]
\[p_{v_1} = p_{v_2} = 2 - 2X_{v_1} - X_l - X_t = 2 - 2 \times \frac{3}{92} - \frac{3}{46} - \frac{4}{23} = \frac{39}{23}.\]

We can continue the price calculation by a topological order. Since there are no multiple buyers case later on, the flow to the downstream is just the sum of the inflow from upstream.

\end{example}

\begin{example}[SPG with a Shortcut] \label{app:ex:shortcut}
\mbox{}\\

Consider the following network where $st$ is a shortcut of path $(s,v,t)$.

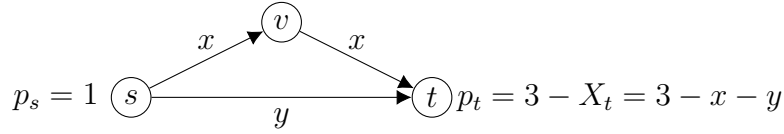
\begin{center}
\begin{tikzpicture}
	\node[vertex] (s) at (0,0) {$s$}; 
	\node at (-1, 0) {$p_s=1$};
	\node[vertex] (v) at (2,1) {$v$}; 
	\node[vertex] (t) at (4,0) {$t$}; 
	\node at (6.5, 0) {$p_{t} = 3 - X_t = 3 - x - y$};
	\path[->]
		(s) edge node [above] {$x$} (v) 
		(v) edge node [above] {$x$} (t) 
		(s) edge node [below] {$y$} (t);
\end{tikzpicture}
\end{center}

At equilibrium, suppose $s$ offers $v$ price $p_{sv}$, and let $x=x_{sv}=x_{vt}$ and $y=x_{st}$. The utility of $v$ is
\begin{align*}
\Pi_v = (3 - x - y) x - p_{sv} x.
\end{align*}

Take the derivative of $\Pi_v$ with respect to $x$, since $\Pi_v$ is concave, the best $p_{sv}$ satisfies:
\begin{align*}
\frac{\partial \Pi_v}{\partial x} = 3-2x-y - p_{sv} = 0 \Rightarrow p_{sv} = 3 - 2x - y.
\end{align*}

The utility of $s$ is
\begin{align*}
\Pi_s &= p_{sv} x + p_t y - p_s(x+y) \\
	&= (3-2x-y)x + (3-x-y)y - (x+y) \\
	&= 2x + 2y - 2x^2 - y^2 - 2xy.
\end{align*}

Take the derivative of $\Pi_s$ with respect to $x$ and $y$, since $\Pi_s$ is concave, we have:
\begin{align*}
\frac{\partial \Pi_s}{\partial x}  &= 2-4x-2y = 0, \\
\frac{\partial \Pi_s}{\partial y}  &= 2-2x-2y = 0.
\end{align*}

The equilibrium solution is $x=0$ and $y=1$. $sv$ and $vt$ are inactive.

\end{example}

\begin{example}[Multiple Equilibria in SMSPG] \label{app: ex: me_smspg}
\mbox{}\\
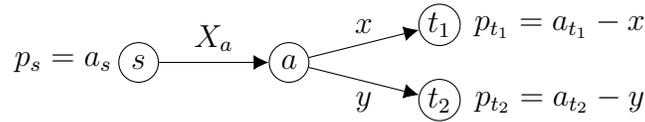
\begin{center}
\begin{tikzpicture}
	\node[vertex] (b) at (0,0) {$s$}; 
	\node at (-1, 0) {$p_s = a_s$};
	\node[vertex] (a) at (2,0) {$a$}; 
	\node[vertex] (t1) at (4,.5) {$t_1$}; 
	\node[vertex] (t2) at (4,-.5) {$t_2$}; 
	\node at (5.6, .5) {$p_{t_1} = a_{t_1} - x$};
	\node at (5.6, -.5) {$p_{t_2} = a_{t_2} - y$};
	\path[->]
		(b) edge node [above] {$X_a$} (a) 
		(a) edge node [above] {$x$} (t1) 
		(a) edge node [below] {$y$} (t2);
\end{tikzpicture}
\end{center}

Suppose $a_{t_1} > a_{t_2} > a_s$. The high price strategy for $s$ is such that $x>0$ and $y=0$ when $a$ tries to maximize its utility $\Pi_a=(a_{t_1}-x-p_a)x$, so
\[\frac{\partial \Pi_a}{\partial x} = a_{t_1} - 2x - p_a = 0 \implies p_a = a_{t_1} - 2x.\]
$s$ tries to maximize its utility $\Pi_s=(a_{t_1}-2x-a_s)x$ given the above $p_a$, so
\[\frac{\partial \Pi_s}{\partial x} = a_{t_1} - 4x - p_s = 0 \implies p_s = a_{t_1} - 4x = a_s \implies X_a = x = \frac{a_{t_1}-a_s}{4},\]
and the high price utility of $s$ is
\[\Pi_s = (p_a - p_s) X_a = (a_{t_1} - 2X_a - a_s)\frac{a_{t_1}-a_s}{4} = \frac{(a_{t_1}-a_s)^2}{8}.\]

The low price strategy for $s$ is such that $x>0$ and $y>0$ when $a$ tries to maximize its utility $\Pi_a=(a_{t_1}-x-p_a)x + (a_{t_2}-y-p_a)y$, so
\[
\begin{cases}
\frac{\partial \Pi_a}{\partial x} = a_{t_1} - 2x - p_a = 0, \\
\frac{\partial \Pi_a}{\partial y} = a_{t_2} - 2y - p_a = 0,
\end{cases}
\implies
p_a = \frac{a_{t_1}+a_{t_2}}{2} - (x+y).
\]
$s$ tries to maximize its utility $\Pi_s=(\frac{a_{t_1}+a_{t_2}}{2} - (x+y)-a_s)(x+y)$ given the above $p_a$, so
\[
\frac{\partial \Pi_s}{\partial x+y} = \frac{a_{t_1}+a_{t_2}}{2} - 2(x+y) - p_s = 0
\]
which implies
\[p_s = \frac{a_{t_1}+a_{t_2}}{2} - 2(x+y) = a_s \implies X_a = x+y = \frac{a_{t_1}+a_{t_2}-2a_s}{4},\]
and the low price utility of $s$ is
\[\Pi_s = (p_a-a_s) X_a = (\frac{a_{t_1}+a_{t_2}}{2} - X_a - a_s)\frac{a_{t_1}+a_{t_2}-2a_s}{4} = \frac{(a_{t_1}+a_{t_2}-2a_s)^2}{16}.\]

The high price strategy and the low price strategy give the same utility to $s$ when

\[\frac{(a_{t_1}-a_s)^2}{8} = \frac{(a_{t_1}+a_{t_2}-2a_s)^2}{16} \implies a_{t_1} = (1+\sqrt{2})a_{t_2} - \sqrt{2}a_s.\]

For price feasibility for the low price strategy, we must have
\[a_{t_2} \geqslant p_a \implies a_{t_2} \geqslant \frac{a_{t_1}+a_{t_2}}{2} - (x+y) = \frac{a_{t_1}+a_{t_2}+2a_s}{4} \implies 3a_{t_2} \geqslant a_{t_1}+2a_s\]
which is feasible when the high price utility and the low price utility are the same for $s$ since
\begin{align*}
a_{t_1} + 2a_s &= (1+\sqrt{2})a_{t_2} - \sqrt{2}a_s + 2a_s \\
&= (1+\sqrt{2})a_{t_2} + (2-\sqrt{2})a_s \\
&\leqslant (1+\sqrt{2})a_{t_2} + (2-\sqrt{2})a_{t_2} \\
&= 3 a_{t_2}.
\end{align*}

There are multiple equilibria when $a_{t_1} = (1+\sqrt{2})a_{t_2} - \sqrt{a_s}$ since $s$ can play either the high price or the low price strategy. We note that in the left upper figure in Figure~\ref{fig: 5.2}, $s$ has no preference when $a_{t_1} = 19.07 = (1+\sqrt{2})a_{t_2} -\sqrt{2} a_s = (1+\sqrt{2})12 - 7\sqrt{2}=12+5\sqrt{2}$ and by fixing $a_{t_2}=12$, $(12, 3\times12 - 2\times7] = (12, 22]$ is the low price feasible interval for $a_{t_1}$.
\end{example}

\begin{example}[SMSPG without an Equilibrium] \label{ex: SMSPG}
\mbox{}\\
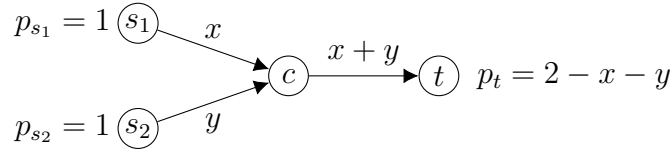
\begin{center}
\begin{tikzpicture}
	\node[vertex] (s1) at (2,.7) {$s_1$};
	\node[vertex] (s2) at (2,-.7) {$s_2$};
	\node at (1, .7) {$p_{s_1} = 1$};
	\node at (1, -.7) {$p_{s_2} = 1$};
	\node[vertex] (c) at (4,0) {$c$}; 
	\node[vertex] (t1) at (6,0) {$t$};
	\node at (7.8, 0) {$p_t = 2-x-y$};
	\path[->]
		(s1) edge node [above] {$x$} (c)
		(s2) edge node [below] {$y$} (c)
		(c) edge node [above] {$x+y$} (t1);
\end{tikzpicture}
\end{center}

Suppose there is an equilibrium where $s_1$ offers the price $p_{s_1 c}$ to $c$ and $s_2$ offers the price $p_{s_2 c}$ to $c$, then $c$ tries to maximize its utility
\[\Pi_c = (2-x-y)(x+y) - p_{s_1 c} x - p_{s_2 c} y\]
where
\[\frac{\partial \Pi_c}{\partial x} = 2 - 2(x+y) - p_{s_1 c} \quad \text{ and } \quad \frac{\partial \Pi_c}{\partial y} = 2 - 2(x+y) - p_{s_2 c}.\]
The equilibrium is a solution of the following LCP:
\[
\begin{cases}
(2-2x-2y-p_{s_1 c}) x + (2-2x-2y-p_{s_2 c}) y = 0,\\
2-2x-2y-p_{s_1 c} \leqslant 0,\\
2-2x-2y-p_{s_2 c} \leqslant 0,\\
x \geqslant 0, \\
y \geqslant 0.
\end{cases}
\]

Now fix the price $p_{s_1 c} \in (1,2)$. To maximize the utility of $s_2$, the best response is to set $p_{s_2 c} = p_{s_1 c} - \varepsilon$ so that $c$ does not buy anything from $s_1$. $s_1$ can have a similar best response to the strategy of $s_2$. $s_1$ and $s_2$ will just set their selling price as close as possible to 1.
\end{example}

\begin{example}[Equilibrium in General DAGs]\label{app: ex: gen_dag}
\mbox{}\\
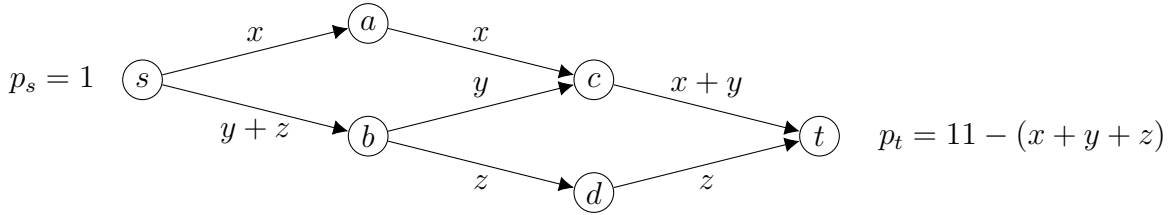
\begin{center}
\begin{tikzpicture}[scale=1.5]
	\node[vertex] (s) at (0,.5) {$s$}; 
	\node at (-0.8, 0.5) {$p_s = 1$};
	\node[vertex] (a) at (2,1) {$a$}; 
	\node[vertex] (b) at (2,0) {$b$}; 
	\node[vertex] (c) at (4,.5) {$c$}; 
	\node[vertex] (d) at (4,-.5) {$d$}; 
	\node[vertex] (t) at (6,0) {$t$}; 
	\node at (7.8, 0) {$p_t = 11 - (x+y+z)$};
	\path[->]
		(s) edge node [above] {$x$} (a) 
		(s) edge node [below] {$y+z$} (b) 
		(a) edge node [above] {$x$} (c) 
		(b) edge node [above] {$y$} (c) 
		(b) edge node [below] {$z$} (d) 
		(c) edge node [above] {$x+y$} (t) 
		(d) edge node [below] {$z$} (t);
\end{tikzpicture}
\end{center}

We compute the price function from $t$ in this network and consider the following cases:
\begin{itemize}
\item $x>0$ and $y+z=0$:
In this case, it is equivalent to consider the line network:
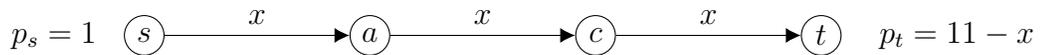
\begin{center}
\begin{tikzpicture}[scale=1.5]
	\node[vertex] (s) at (0,0) {$s$}; 
	\node at (-0.8, 0) {$p_s = 1$};
	\node[vertex] (b) at (2,0) {$a$};
	\node[vertex] (d) at (4,0) {$c$}; 
	\node[vertex] (t) at (6,0) {$t$}; 
	\node at (7.2, 0) {$p_t = 11 - x$};
	\path[->]
		(s) edge node [above] {$x$} (b) 
		(b) edge node [above] {$x$} (d) 
		(d) edge node [above] {$x$} (t);
\end{tikzpicture}
\end{center}
We have
\[p_c = 11-2x, \quad p_a = 11-4x, \text{ and}\]
\[p_s = 11-8x = 1 \implies x=\frac{5}{4}.\]
The utility of $s$ is
\[(p_a - p_s)x = 4x^2 = \frac{25}{4}.\]

\item $x=0$ and $y+z>0$:
In this case, the best strategy for $s$ is to make $y>0$ and $z>0$, so it is equivalent to consider the following network:

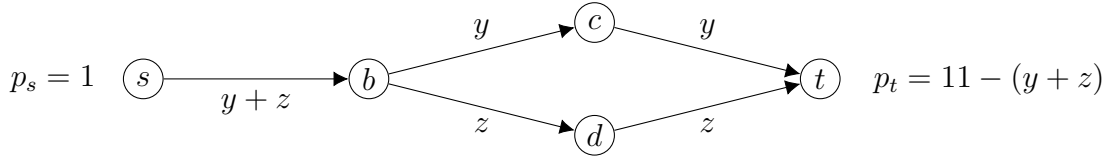
\begin{center}
\begin{tikzpicture}[scale=1.5]
	\node[vertex] (s) at (0,0) {$s$}; 
	\node at (-0.8, 0) {$p_s = 1$};
	\node[vertex] (b) at (2,0) {$b$}; 
	\node[vertex] (c) at (4,.5) {$c$}; 
	\node[vertex] (d) at (4,-.5) {$d$}; 
	\node[vertex] (t) at (6,0) {$t$}; 
	\node at (7.5, 0) {$p_t = 11 - (y+z)$};
	\path[->]
		(s) edge node [below] {$y+z$} (b) 
		(b) edge node [above] {$y$} (c) 
		(b) edge node [below] {$z$} (d) 
		(c) edge node [above] {$y$} (t) 
		(d) edge node [below] {$z$} (t);
\end{tikzpicture}
\end{center}
By Algorithm~\ref{alg: 1} and \ref{alg: 2}, we have
\[p_c = 11-2y-z, \quad p_d = 11-y-2z, \quad p_b = 11-3(y+z),\]
\[\text{and } p_s = 11-6(y+z) = 1 \implies y+z=\frac{5}{3}.\]
The utility of $s$ is
\[(p_b - p_s)(y+z) = 3(y+z)^2 = \frac{25}{3}.\]

\item $x>0$ and $y+z>0$: In this case, if the equilibrium is equivalent to consider the following network:
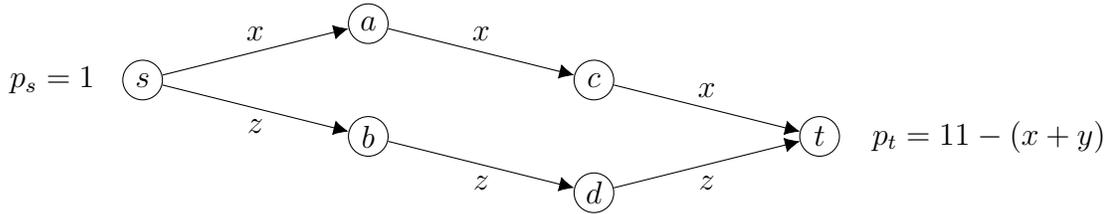
\begin{center}
\begin{tikzpicture}[scale=1.5]
	\node[vertex] (s) at (0,.5) {$s$}; 
	\node at (-0.8, 0.5) {$p_s = 1$};
	\node[vertex] (a) at (2,1) {$a$}; 
	\node[vertex] (b) at (2,0) {$b$}; 
	\node[vertex] (c) at (4,.5) {$c$}; 
	\node[vertex] (d) at (4,-.5) {$d$}; 
	\node[vertex] (t) at (6,0) {$t$}; 
	\node at (7.5, 0) {$p_t = 11 - (x+y)$};
	\path[->]
		(s) edge node [above] {$x$} (a) 
		(s) edge node [below] {$z$} (b) 
		(a) edge node [above] {$x$} (c) 
		(b) edge node [below] {$z$} (d) 
		(c) edge node [above] {$x$} (t) 
		(d) edge node [below] {$z$} (t);
\end{tikzpicture}
\end{center}
then by Algorithm~\ref{alg: 1} and \ref{alg: 2}, we have
\[p_a = 11-4x-z, \quad p_b = 11-x-4z,\]
\[p_s = 11-5(x+z) = 1 \implies x+z=2, \text{ and } x=z=1.\]
The utility of $s$ is
\[(p_a - p_s)x + (p_b - p_s)z = 5 + 5 = 10.\]
Besides, $p_{ac} = 11-2x-z = p_{bd} = 11-x-2z = 8$.

This is indeed an equilibrium. Given that $p_{ac} = 8$, if $b$ wants to earn profit from $c$, then $b$ must set the price $p_{bc} \leqslant 8$. However, $b$ does not have the incentive to compete with $a$ because $b$ can sell all the goods with $z=1$ to $d$ at price $p_{bd}=8$.

\end{itemize}
The equilibrium is the last case when $x=z=1$, $y=0$, $p_{sa} = p_{sb} = 6$, $p_{ac} = 11 - 2x - z = 8 = 11 - x - 2z = p_{bd}$, and $p_{bc}$ can be any positive number.
\end{example}

\end{document}